\documentclass[11pt]{article}
\usepackage[T1]{fontenc}
\usepackage[utf8]{inputenc}
\usepackage{lmodern}
\usepackage[margin=1in]{geometry}
\usepackage{amsmath,amssymb,amsthm,mathtools}
\usepackage{microtype}
\usepackage{booktabs,longtable,array}
\usepackage[colorlinks=true,allcolors=blue,bookmarksnumbered=true]{hyperref}
\hypersetup{
  pdftitle={An exponential lower bound for the bit pigeonhole principle in resolution over parities},
  pdfauthor={Kamil Braun},
  pdfsubject={Revised research preprint; exposition and illustrations updated},
  pdfkeywords={proof complexity, bit pigeonhole principle, resolution over parities, polynomial calculus}
}
\newtheorem{theorem}{Theorem}[section]
\newtheorem{lemma}[theorem]{Lemma}
\newtheorem{corollary}[theorem]{Corollary}
\newtheorem{proposition}[theorem]{Proposition}
\theoremstyle{definition}

\theoremstyle{remark}
\newtheorem{remark}[theorem]{Remark}
\numberwithin{equation}{section}
\newcommand{\leanref}[2]{\href{https://github.com/kbr-/math-research/blob/8904bf09a5376f48a00d1c25d079bf0c6441502a/formalization/#1}{#2}}
\newcommand{\lean}[1]{\textnormal{\small\leanref{#1}{[Lean]}}}
\newcommand{\leangroup}[1]{\textnormal{\small[Lean: #1]}}
\newcommand{\leansep}{\,\textperiodcentered\ }
\newcommand{\F}{\mathbb F_2}
\newcommand{\Res}{\mathrm{Res}(\oplus)}

\newcommand{\BPHP}{\mathrm{BPHP}}
\newcommand{\Bool}{\operatorname{BOOL}}
\newcommand{\Ideal}{\mathcal I}
\newcommand{\Cons}{\mathcal C}
\newcommand{\Pol}{\mathcal P}
\newcommand{\rows}{\operatorname{row}}
\newcommand{\cols}{\operatorname{col}}
\newcommand{\Span}{\operatorname{span}}

\title{An exponential lower bound for the bit pigeonhole principle\\
in resolution over parities}
\author{Kamil Braun\\
\href{mailto:kamilgbraun@gmail.com}{\texttt{kamilgbraun@gmail.com}}\\[0.6em]
\small With substantial assistance from Claude and GPT models.\\
\small More details in Section~\ref{sec:model-assistance}.}
\date{Revised manuscript\\22 October 2026}
\usepackage{xcolor,tikz,needspace}
\usetikzlibrary{arrows.meta,positioning}
\definecolor{ink}{HTML}{173F4A}
\definecolor{sea}{HTML}{287B83}
\definecolor{mist}{HTML}{EDF5F4}
\definecolor{quiet}{HTML}{7A8589}
\hypersetup{allcolors=sea}
\makeatletter
\renewcommand{\@maketitle}{%
 \newpage
 \begin{center}
  {\Large\bfseries\color{ink}\@title\par}
  \vskip 1.1em
  {\normalsize\@author\par}
  \vskip .8em
  {\small\@date\par}
 \end{center}
 \par\vskip 1.2em}
\renewenvironment{abstract}{%
 \begin{center}\small\bfseries Abstract.\end{center}
 \begin{quotation}\small\noindent
}{\par\end{quotation}\vspace{.6em}}
\makeatother

\usepackage{fancyhdr,placeins,float}
\definecolor{rim}{HTML}{B45C25}
\definecolor{conflict}{HTML}{A43B38}
\newcommand{\journeygraph}[7]{%
\begin{tikzpicture}[x=1.05cm,y=.78cm,font=\normalsize,text=ink,
 v/.style={circle,fill=ink,minimum size=4.5pt,inner sep=0pt},
 arr/.style={draw=sea,line width=.8pt,-{Stealth[length=1.4mm]},shorten <=2pt,shorten >=2pt},
 tree/.style={draw=sea!80,line width=.8pt,densely dashed,-{Stealth[length=1.4mm]},shorten <=2pt,shorten >=2pt}]
\node[v,label=above:{$#1$}](c)at(0,2.9){};
\node[v,label=left:{$#2$}](d0)at(-1.15,1.65){};
\node[v,label=right:{$#3$}](d1)at(1.15,1.65){};
\node[v,label=below left:{$#4$}](l0)at(-1.15,.15){};
\node[v,label=below right:{$#5$}](l1)at(1.15,.15){};
\node[v,label=below right:{$#6$}](l)at(0,-1.0){};
\node[v,label=below:{$#7$}](end)at(0,-2.3){};
\draw[arr](c)to[bend right=19](d0);\draw[arr](c)to[bend left=19](d1);
\draw[tree](d0)to[bend right=15](l0);
\draw[tree](d1)to[bend left=15](l1);
\draw[arr](l0)to[bend right=18](l);\draw[arr](l1)to[bend left=18](l);
\draw[tree](l)to[bend right=20](end);
\draw[tree,draw=gray!80](-1.65,.95) .. controls(-1.62,.70)and(-1.30,.58) .. (l0);
\draw[tree,draw=gray!80](1.65,.95) .. controls(1.62,.70)and(1.30,.58) .. (l1);
\draw[tree,draw=gray!80](.45,-1.85) .. controls(.40,-2.0)and(.16,-2.08) .. (end);
\end{tikzpicture}}

\newcommand{\evalpoly}[1]{\begin{gathered}#1\\[-2pt]\textcolor{sea}{\scriptstyle\Lambda(\cdot)=0}\end{gathered}}

\newcommand{\surfacepanel}[1]{%
\begin{tikzpicture}[x=.85cm,y=.85cm,font=\small,text=ink]
\coordinate(A)at(-2,0);\coordinate(B)at(-1.25,-.85);\coordinate(K)at(.85,-1.05);\coordinate(E)at(2,0);\coordinate(U)at(0,2.4);
\ifnum#1=1
\filldraw[fill=sea!15,draw=sea] (U)--(A)--(B)--cycle;
\filldraw[fill=sea!30,draw=sea] (U)--(B)--(K)--cycle;
\filldraw[fill=sea!20,draw=sea] (U)--(K)--(E)--cycle;
\fi
\draw[rim,line width=1.3pt](A)--(B)--(K)--(E)--(U)--cycle;
\foreach \v in {A,E,U}{\fill[ink](\v)circle(2.7pt);}
\foreach \v in {B,K}{\filldraw[fill=white,draw=ink](\v)circle(2.7pt);}
\node[left]at(A){$A$};\node[below left]at(B){$B$};\node[below right]at(K){$K$};\node[right]at(E){$E$};\node[above]at(U){$U$};
\end{tikzpicture}}

\newcommand{\sourceonlyfigure}{%
\begin{figure}[H]\centering
\journeygraph{C}{C_1}{C_2}{H_1}{H_2}{H}{\bot}
\caption{A schematic refutation with a reused clause $C$. Dashed paths
abbreviate subderivations; gray arrows supply other premises. One reuse
is shown, with no restriction on the structure of the rest of the proof.}
\label{fig:source-dag}
\end{figure}}

\newcommand{\pconlyfigure}{%
\begin{figure}[htbp]\centering
\journeygraph{P_C}{P_{C_1}}{P_{C_2}}{P_{H_1}}{P_{H_2}}{P_H}{1}
\caption{The clause values in Figure~\ref{fig:source-dag} after PC
simulation. The shared line $P_C$ is derived once and reused. Each source
inference expands into local PC steps, all through the same degree
$D=\max\{2h+\ell,4h+1\}$, including the suppressed lines.}
\label{fig:pc-dag}
\end{figure}}

\newcommand{\weightedonlyfigure}{%
\begin{figure}[htbp]\centering
\journeygraph{f\phi(P_C)}{f\phi(P_{C_1})}{f\phi(P_{C_2})}{f\phi(P_{H_1})}{f\phi(P_{H_2})}{f\phi(P_H)}{f}
\caption{Weighted removal sends the final line $1$ in
Figure~\ref{fig:pc-dag} to $f$. All extension variables have disappeared;
every line, including those on dashed paths, is derived from the old base
through degree $B=k(D+1)$.}
\label{fig:weighted-dag}
\end{figure}}

\newcommand{\momentfigure}{%
\begin{figure}[htbp]\centering
\begin{minipage}{.47\linewidth}\centering\surfacepanel{0}\end{minipage}\hfill
\begin{minipage}{.47\linewidth}\centering\surfacepanel{1}\end{minipage}
\caption{Pair moments prescribe the orange boundary (left). Three
matching triangles fill it (right): $c=UAB+UBK+UKE$ has
$\partial c=AB+BK+KE+EU+UA$. The orange edges and shaded triangles have coefficient
one; filled and hollow vertices have singleton moments one and zero.
The two internal edges cancel over $\F$. Coordinates are
$A=(2,\alpha)$, $B=(3,\beta)$, $K=(2,\gamma)$,
$E=(3,\alpha)$, $U=(4,\delta)$, with four distinct columns.
Only the prescribed boundary and filler are drawn, not the whole complex.}
\label{fig:moment-filling}
\end{figure}}

\newcommand{\contradictionfigure}{%
\begin{figure}[htbp]\centering
\begin{tikzpicture}[text=ink,font=\small]
\node at (-3.5,3.7) {The assumed derivation};
\node at (3.5,3.7) {Coefficient isolation and moment extension};
\node at (-3.5,.35) {\journeygraph{\evalpoly{f\phi(P_C)}}{\evalpoly{f\phi(P_{C_1})}}{\evalpoly{f\phi(P_{C_2})}}{\evalpoly{f\phi(P_{H_1})}}{\evalpoly{f\phi(P_{H_2})}}{\evalpoly{f\phi(P_H)}}{f}};
\node at (3.5,.5) {\surfacepanel{1}};
\node[font=\scriptsize,text=quiet] at (3.5,-1.9) {Illustrative moment filling};
\node[align=center] (zero) at (-3.5,-3.25) {$f\in\Cons_B(\mathcal Q_{m,\ell})$\\\textcolor{sea}{$\Lambda(f)=0$}};
\node[align=center] (one) at (3.5,-3.25) {$\Delta\tau f=1,\quad\lambda(1)=1$\\\textcolor{rim}{$\Lambda(f)=1$}};
\node[text=conflict,font=\normalsize,inner sep=2pt] (equation) at (0,-4.55) {$0=\Lambda(f)=1$};
\draw[-{Stealth[length=1.7mm]},sea,shorten >=2pt] (zero.south)--([xshift=4pt]equation.north west);
\draw[-{Stealth[length=1.7mm]},rim,shorten >=2pt] (one.south)--([xshift=-4pt]equation.north east);
\end{tikzpicture}
\caption{The same polynomial and the same functional give incompatible
conclusions. The bounded derivation forces $\Lambda(f)=0$; coefficient
isolation and a normalized matching functional give $\Lambda(f)=1$.
The surface illustrates the filling step used to construct that functional.}
\label{fig:contradiction}
\end{figure}}

\newcommand{\covercorrectiondiagram}{%
\begin{center}
\begin{tikzpicture}[font=\small,text=ink,
 map/.style={-{Stealth[length=1.6mm]},draw=sea,line width=.65pt}]
\node[align=center] (d) at (0,1.4) {$D_{a,b}$\\$d$};
\node[align=center] (x) at (3.4,1.4) {$D_{a+1,b}$\\$x,\ Hz$};
\node[align=center] (z) at (7,1.4) {$D_{a+2,b}$\\$z$};
\node[align=center] (vx) at (3.4,-.5) {$D_{a+1,b-1}$\\$Vx=Hy$};
\node[align=center] (y) at (7,-.5) {$D_{a+2,b-1}$\\$y=Vz$};
\draw[map] (x)--node[above] {$H$}(d);
\draw[map] (z)--node[above] {$H$}(x);
\draw[map] (x)--node[left] {$V$}(vx);
\draw[map] (z)--node[right] {$V$}(y);
\draw[map] (y)--node[above] {$H$}(vx);
\end{tikzpicture}
\end{center}}

\usepackage{etoolbox}
\pretocmd{\section}{\thispagestyle{plain}}{}{}

\begin{document}
\maketitle
\thispagestyle{plain}
\begin{abstract}
Resolution over parities, $\Res$, reasons with disjunctions of affine equations
over $\F$. We prove that every DAG-like $\Res$ refutation of the bit pigeonhole
principle with $n+1$ pigeons and $n=2^\ell$ holes has more than
$\exp\bigl(n/(32768\ell^2)\bigr)=2^{\Omega(n/\log^2 n)}$ proof-lines for
$\ell\ge32$, without restrictions on regularity or depth.

The proof translates a refutation with $S$ lines into a polynomial-calculus
refutation of degree $O(\log n)$ using $O(S+n^2)$ blocks of extension variables.
Removing these variables is the central step. Low-rank blocks admit an exact
substitution; high-rank blocks admit a simultaneous weighted substitution,
using one nonzero polynomial in the common kernel of their restriction maps.
The result is a nonzero row-linear consequence of the original pigeonhole
axioms at degree at most $n/2$. A separating functional, constructed through
matching moments and chessboard homology, excludes such a consequence.
The argument also gives a sufficient condition for lower bounds on other
linear-clause systems.

The main results and their mathematical dependencies are formalized in Lean~4;
the paper includes a verification map and reproduction instructions. The
mathematics was developed with substantial AI assistance in an open research
framework; Section~\ref{sec:framework} describes its provenance and contributions.
\end{abstract}
\clearpage
\section{The claim and its setting}\label{sec:introduction}

With more pigeons than holes, two pigeons must share a hole. Knowing that a
collision exists, however, does not tell us how much work a proof system must
do to expose it. In the bit encoding, each pigeon chooses a hole by its binary
label. Resolution over parities can reason about affine relations among many
of these bits at once. A DAG-like proof can also reuse every intermediate
clause. Can an exponential lower bound survive both freedoms?

We prove that it does. For $n=2^\ell$ holes and $n+1$ pigeons, every refutation
has more than $\exp(n/(32768\ell^2))$ proof-lines when $\ell\ge32$, without a
regularity or depth restriction. The proof first uses extension variables to
obtain a low-degree polynomial-calculus refutation. Its central task is then
to remove those variables. A single low-degree polynomial, chosen in the
common kernel of many restriction maps, makes that removal possible.

Resolution over linear equations was introduced by Raz and Tzameret
\cite{RT08} over the integers. Its characteristic-two version, resolution
over parities, was subsequently studied by Itsykson and Sokolov
\cite{IS14,IS20} in the tree-like setting. The unrestricted size problem
remained an explicit benchmark in the literature preceding this result
\cite{IPS26}.

\subsection{The bit pigeonhole principle}
Fix \(\ell\ge2\), put \(n=2^\ell\) and \(m=n+1\), and introduce
\[
  b_{it}\qquad (i\in[m],\ t\in[\ell]).
\]
The row \(b_i\) encodes one of the \(n\) labels in \(\{0,1\}^{\ell}\).
For a label \(z\), let \([b_{it}\ne z_t]\) denote \(b_{it}\) when \(z_t=0\)
and \(\neg b_{it}\) when \(z_t=1\). The usual CNF encoding is
\begin{equation}\label{eq:bit-php}
 \BPHP_n^{n+1}
  =\bigwedge_{\substack{i<i'\in[m]\\z\in\{0,1\}^{\ell}}}
    \left(\bigvee_{t\in[\ell]}[b_{it}\ne z_t]
       \ \vee\ \bigvee_{t\in[\ell]}[b_{i't}\ne z_t]\right).
\end{equation}
It has \(v=m\ell\) variables and \(n\binom m2\) clauses of width \(2\ell\).
No separate range or pigeon axiom is required: every Boolean row encodes
exactly one available label.

\subsection{The proof system}
A \emph{linear clause} is a disjunction of affine equations over \(\F\).
The basic rules are complementary-parity resolution and semantic weakening:
\begin{equation}\label{eq:res-rules}
 \frac{A\vee(u=0)\qquad B\vee(u=1)}{A\vee B},
 \qquad
 \frac{C}{D}\quad\text{if } C\models D.
\end{equation}
The pivot \(u\) may be any affine form. A refutation is a finite acyclic
derivation from initial clauses to the empty clause. Its size is the number
of proof-lines, including initial lines. Equivalently, it counts the nodes
of a topologically ordered proof DAG. Earlier proof-lines may be
used repeatedly, and no restriction is placed on proof depth or on repeated
use of related pivots.

We also consider the convention permitting every sound inference from at
most two linear-clause premises. Lemma~\ref{lem:affine-cover} reduces this
convention to \eqref{eq:res-rules} with constant proof-line overhead.
Fresh unconstrained propositional variables, if used, may be set to
constants throughout a proof. Extension axioms defining new predicates
are not part of the system considered here.

\begin{samepage}
\begin{theorem}[Main theorem]\label{thm:main}\lean{claims/BitPHPExponential.lean}
For every \(\ell\ge32\), with \(n=2^\ell\), every DAG-like \(\Res\)
refutation of \(\BPHP_{n}^{n+1}\) has more than
\[
 \exp\!\left(\frac{n}{32768\,\ell^2}\right)
\]
proof-lines. The same bound holds with base two in place of \(e\), and for both
rule conventions above.
\end{theorem}
\end{samepage}

In terms of \(n\) the bound is \(2^{\Omega(n/\log^2n)}\). We call it
exponential in the usual sense of proof complexity: the formula
\eqref{eq:bit-php} has size \(L=\Theta(n^3\log n)\), and the bound is
\(2^{L^{1/3-o(1)}}\). No bound of the form \(2^{\Omega(L)}\) or
\(2^{\Omega(n)}\) is claimed. The constant
\(32768=2^{15}\) and the threshold \(32\) are what the formal proof
verifies; no attempt was made to optimize them. The following qualitative consequence also has an independent proof.

\begin{samepage}
\begin{corollary}[Superpolynomial form]\label{cor:main-superpolynomial}
\lean{claims/BitPHPSuperpolynomial.lean}
For every fixed real \(K>0\), there is \(\ell_0(K)\) such that for every
\(\ell\ge\ell_0(K)\), every DAG-like \(\Res\) refutation of
\(\BPHP_{2^\ell}^{2^\ell+1}\) has more than \(2^{K\ell}=n^K\) proof-lines.
\end{corollary}
\end{samepage}

The input length in \eqref{eq:bit-php} is polynomial in \(n\), so the bound
is also exponential in a fixed power of the input length, and it bounds the
total proof-description size from below.

\subsection{Prior results and the scope of the claim}
Efremenko, Garl\'ik, and Itsykson \cite{EGI25} prove a
\(2^{\Omega(n^{1/3}/\log n)}\) size bound for regular \(\Res\) refutations
of bit PHP. Their unrestricted DAG-like result is an affine-width bound.
Bhattacharya, Chattopadhyay, and Dvo\v{r}\'ak \cite{BCD24} exhibit formulas
that are exponentially hard for bottom-regular \(\Res\), although they
have short proofs in general ordinary resolution. This separation
demonstrates the importance of removing regularity.

Bhattacharya and Chattopadhyay \cite{BC25}, Byramji and Impagliazzo
\cite{BI25}, and the merged STOC 2026 paper \cite{BCBI26} obtain
exponential bounds in regimes allowing nearly quadratic proof depth.
Their bounds retain a proof-depth restriction. Itsykson, Podolskii,
and Shekhovtsov \cite{IPS26} explicitly distinguish this progress from
the general superpolynomial size problem; their unrestricted size
consequence is quadratic.

Alekseev and Gaevoy \cite{AG26} study constrained bit PHP. Their
polynomial-depth bounds for general \(\Res\) are conditional; their
unconditional results retain either a reversible-system restriction or
a near-quadratic depth bound. They therefore do not give the unrestricted
usual-bit-PHP size conclusion considered here.

\paragraph{Resolution over linear equations.}
The family \(\mathrm{Res}(\mathrm{lin}_R)\) studies resolution over linear
equations in a ring \(R\). Raz and Tzameret \cite{RT08} introduced the
integer formulation; Itsykson and Sokolov \cite{IS14,IS20} studied the
\(\F\) version in the tree-like setting. Part and Tzameret
\cite{PT21} proved the first superpolynomial lower bounds for DAG-like
\(\mathrm{Res}(\mathrm{lin}_R)\): the subset-sum principle requires
exponential size over the rationals. Their technique uses large
characteristic and does not apply over \(\F\). Khaniki \cite{Kha22} proved
almost quadratic lower bounds for DAG-like resolution over low-degree
polynomial equations over finite fields, which includes resolution over
linear equations, for mod-\(q\) Tseitin formulas and random \(k\)-CNFs.
To our knowledge these were the first superlinear DAG-like bounds for such
systems over finite fields; \cite{EGI25} notes that the rules considered
there differ from the formulation of \(\Res\) used here.

\paragraph{Other work on \(\Res\).}
Gryaznov, Pudl\'ak, and Talebanfard \cite{GPT22} introduced regular \(\Res\)
and related it to read-once linear branching programs. Alekseev and Itsykson
\cite{AI25} lift resolution depth to regular and bounded-depth \(\Res\) size,
and Efremenko and Itsykson \cite{EI25,EI26} sharpen the closure technique
of \cite{EGI25} and obtain strong bounds for bounded-depth \(\Res\). All of
these retain a regularity or depth restriction.

\paragraph{Polynomial calculus and the pigeonhole principle.}
The degree lower bound at the base of our argument is Razborov's theorem
\cite{Raz98} that polynomial calculus refutations of the pigeonhole principle
need degree \(n/2+1\) over every field, which also describes the low-degree
consequences of the pigeonhole axioms. Related degree bounds and methods are
due to Impagliazzo, Pudl\'ak, and Sgall \cite{IPS99}, Alekhnovich and Razborov
\cite{AR03}, and Mik\v{s}a and Nordstr\"om \cite{MN15}. Section~\ref{sec:moments}
reproves the part of Razborov's theorem that we need, over \(\F\), by a
different route through chessboard complexes; we do not claim this degree
range as new (see the discussion before Theorem~\ref{thm:moments}).
The extension variables are those of Buss, Impagliazzo, Kraj\'i\v{c}ek,
Pudl\'ak, Razborov, and Sgall \cite{BIKPRS}.

\paragraph{Scope.}
The theorem here removes both regularity and proof-depth restrictions.
The cited regular and depth-restricted bounds do not cover this setting.
Targeted literature searches on 15 and 17 September 2026, with a further
targeted check of recent records on 22 October 2026, found no earlier theorem
with the same scope; these searches were not exhaustive. The theorem concerns \(\Res\) only. It does not address unary
PHP in \(\Res\), and the \(\mathrm{AC}^0[p]\)-Frege lower-bound problem
remains open; Section~\ref{sec:generic} says what does and does not follow.

\paragraph{Revision.}
The theorem statements and formal sources are unchanged from the September
2026 preprint. This revision develops the exposition and adds worked
illustrations and figures; contributions are described in
Section~\ref{sec:model-assistance}.

\subsection{Organization}
Section~\ref{sec:overview} gives an overview of the proof, fixes terminology
and parameters, and works through one simulated inference.
Sections~\ref{sec:algebra}--\ref{sec:affine} prove the components
in full, and Section~\ref{sec:main-proof} closes the parameters.
Section~\ref{sec:generic} states the generic criterion; it can be skipped.
For a route centered on elimination, Theorem~\ref{thm:row-separation}
can be taken as the separator interface after Section~\ref{sec:algebra}:
Sections~\ref{sec:clause-simulation}--\ref{sec:main-proof} then give the
simulation, removal, and lower bound. Sections~\ref{sec:moments}--\ref{sec:decoder}
and Appendix~\ref{sec:topology} supply that interface's proof.
Appendix~\ref{sec:topology} proves the homological input behind
the matching extension; Appendix~\ref{sec:alternative-proof} gives the
independent smaller-parameter proof of the superpolynomial consequence.
A companion online research notebook records the
development of the argument and supporting experiments \cite{Repo};
Section~\ref{sec:framework} describes these supplementary materials.
The proofs presented here can be read independently of that notebook:
none of its lemmas is required as an unexpanded black box in this preprint.

Every [Lean] link points to a source file in the published revision
\href{https://github.com/kbr-/math-research/tree/8904bf09a5376f48a00d1c25d079bf0c6441502a}{\texttt{8904bf09}}, rather than a moving
branch. The linked files identify their exported statements and dependencies.
Appendix~\ref{sec:verification} maps the main proof steps to exact Lean
declarations and gives the checkout, setup, and verification commands.

\section{Overview of the proof}\label{sec:overview}

Suppose, towards a contradiction, that there is a refutation with $S$
proof-lines, where $S\le\exp(n/(32768\ell^2))$ and $\ell\ge32$.
A direct translation
of a wide clause into a polynomial can have high degree, so the usual degree
lower bound need not contradict it. Extension variables give a much smaller
polynomial representation. They also change the axiom system. A low-degree
refutation in the enlarged system is perfectly compatible with a degree lower
bound for the original one.

The auxiliary variables have earned their keep. Now we have to evict them.
The removal need not preserve the last line. Instead, it will turn the refutation into a derivation of one
nonzero polynomial in the original variables. The base system must exclude
that polynomial as well as the constant one.

\subsection{A low-degree representation of clauses}

Write a linear clause $C$ using affine \emph{true-indicators} $g_i$: the
equation represented by $g_i$ holds exactly when $g_i=1$. Thus
$C=\bigvee_i(g_i=1)$, and its falsifying assignments form the affine set
$Z(C)=\{g_1=\cdots=g_s=0\}$. The ordinary product $\prod_i(1-g_i)$
vanishes exactly where $C$ holds, but its degree can grow with the clause.

Following the extension-variable construction of Buss, Impagliazzo,
Kraj\'i\v{c}ek, Pudl\'ak, Razborov, and Sgall~\cite{BIKPRS}, introduce
fresh Boolean coefficients $r_{ui}$ and put
\[
 P_C=\prod_{u=1}^{h}\left(1-\sum_{i=1}^{s}r_{ui}g_i\right).
\]
The parameter $h$ fixes the number of factors. The degree is at most $2h$,
independently of how many indicators the clause contains. At an assignment
in $Z(C)$, every factor is one. At an assignment satisfying $C$, a uniformly
chosen coefficient row makes its factor zero with probability one half.
Independent rows therefore leave the product equal to one with probability
$2^{-h}$.

The extension axioms are the \emph{companions} $g_iP_C=0$, for every input,
and the coefficient equations $r_{ui}^2-r_{ui}=0$. Together they form a
\emph{complete block}. The companions force the product to be supported on
$Z(C)$; deriving $P_C=0$ then expresses that $C$ has been proved.
It is essential that $P_C=0$ is a derived line rather than another axiom.
For the empty clause the product is one, so its value gives a refutation.

Each inference can be simulated in polynomial calculus (PC) by a small
polynomial identity. In particular,
\[
 1-P_C=\sum_i U_{C,i}g_i,\qquad \deg U_{C,i}\le2h-1,
\]
lets the simulation absorb a clause's side literals using the conclusion's
companions. Resolution costs degree at most $4h+1$; an initial clause costs
at most $2h+\ell$. Earlier premise values are already available as proof-lines.
Multiplying such a line does not require multiplying the derivation that
produced it. Consequently, the degree ceiling is the maximum of the local
costs, not their sum along a path.

This is where a DAG ceases to be a problem. It may be deep, and a line may
be reused many times; neither operation raises the common degree ceiling.
Theorem~\ref{thm:clause-transfer} gives
\[
 D=\max\{2h+\ell,4h+1\},\qquad
 N_{\mathrm{reg}}\le3S+\binom m2.
\]
Here the \emph{registry} is the whole family of blocks, fixed before the
simulation begins. Its slots cover the source lines, possible auxiliary
clauses, and compact initial clauses. With $h=3\ell$, the resulting
refutation has degree $D=12\ell+1$.

The variables $b_{it}$ will be called the \emph{old variables}, to
distinguish them from the new coefficients. The \emph{old base} is the compact pigeonhole system
\begin{equation}\label{eq:compact-base}
 \mathcal Q_{m,\ell}
  =\{E_{ii'}:i<i'\in[m]\}\cup\{b_{it}^2-b_{it}:i\in[m],\ t\in[\ell]\},
 \qquad
 E_{ii'}=\prod_{t=1}^{\ell}(1-b_{it}-b_{i't}).
\end{equation}
On Boolean assignments, $E_{ii'}$ is one exactly when the two rows have
the same label. Thus $E_{ii'}=0$ excludes all their collisions at once.
The compact clause $\bigvee_t(b_{it}+b_{i't}=1)$ implies every label-specific
CNF clause for that pair. Theorem~\ref{thm:clause-transfer} derives those
initial clause values at bounded PC degree. The source refutation is still
of the usual CNF; this compact system is its algebraic base.

Every block input is affine in the old bits; no input uses another
block's coefficients. This one-level structure permits simultaneous
substitution.

\subsection{One parity-resolution step}\label{sec:local-resolution}
For concrete local calculations, take seventeen pigeons and sixteen holes,
with four bits per row. Write $x=b_{11}$, $y=b_{21}$, $z=b_{31}$, and
$w=b_{41}$ for the first bits of four distinct pigeons. The parity
$u=x+y$ can be a resolution pivot just as a single bit can.
At accuracy $h_*=1$, resolve $A=(u=1)$ with $B=(u=0)$ to the empty
clause. Starred parameters will refer to these local calculations; the
assumed global refutation retains the general parameters above.
The opposite premises below illustrate a single inference. The later
packing and removal calculations use a different block family on these
same bit rows. The true-indicators are \(u\) for
\(A\) and \(1-u\) for \(B\), so with coefficients \(r\) and \(s\)
\[
 P_A=1-ru,\qquad P_B=1-s(1-u),\qquad P_\varnothing=1 .
\]
If the premise values $P_A$ and $P_B$ have been derived, the Boolean
identity $u^2-u=(x^2-x)+(y^2-y)$ lets polynomial calculus derive
\[
 (1-u)P_A-r(u^2-u)=(1-u)(1-ru)-r(u^2-u)=1-u ,
\]
which is \eqref{eq:pivot-correction} in this case: from ``\(A\) holds'' we
have derived that the indicator of \(B\) is zero. Multiplying by \(s\) and
adding \(P_B\) gives
\[
 s(1-u)+\bigl(1-s(1-u)\bigr)=1 ,
\]
the value of the empty clause. The simulation uses degree at most three. With side literals present, the same
computation is carried out after multiplying by the product \(P_T\) of the
conclusion \(T\), and the side literals are absorbed by the companions of
\(T\); with \(h\) factors, \(r\) and \(s\) are replaced by the polynomials
\(U_{A,u}\) and \(U_{B,1-u}\) of degree \(2h-1\), which gives the bound
\(4h+1\) of Lemma~\ref{lem:clause-resolution}.

\subsection{The cost of removing the extensions}

Can we replace all the coefficients by low-degree polynomials in the old
variables? The relevant measure of a block is the rank $r$ of its affine
input span. A span containing one has no common zero and can be handled by
a constant substitution. Otherwise its zero set is a flat of codimension
$r$.

For $r\le h(k+1)$, the answer is exact. Split a basis of the inputs among
the $h$ factors, with at most $k+1$ basis elements in each. Telescoping
turns each factor into a product of affine zero-indicators, using coefficient
polynomials of degree at most $k$. The substituted product is the indicator
of the whole zero flat. Its companions vanish on every Boolean assignment
and have bounded-degree Boolean certificates
(Lemma~\ref{lem:packing}).

For the local inputs $x,y$, set $k_*=1$. Their rank is
$2=h_*(k_*+1)$, and one factor can encode them with degree-one coefficients:
\[
 r_x=1,\quad r_y=1-x,\qquad
 1-r_xx-r_yy=1-x-(1-x)y=(1-x)(1-y).
\]
The coefficient of $y$ supplies the extra factor without introducing
another coefficient row. In general, each of the $h$ rows has room for
$k+1$ affine inputs.

High rank exceeds that capacity. But it also leaves fewer free coordinates
on the zero flat. Restriction to such a flat loses much of the information
in a low-degree polynomial. We can exploit that loss instead of representing
the flat's indicator exactly.

Suppose a nonzero polynomial $f$ of degree at most $k$ has zero ordinary
restriction to the flat. In affine coordinates this means the zero
polynomial, not merely a function that vanishes at binary points. Then
Lemma~\ref{lem:restriction-ideal} gives
\[
 f=\sum_i a_i g_i,\qquad \deg a_i\le k-1.
\]
Use these cofactors as the coefficients of the first factor and zero in the
remaining factors. The product becomes $1-f$. A companion becomes
$g_i(1-f)$, which need not vanish. Multiplying it by $f$, however, gives
\[
 fg_i(1-f)=-g_i(f^2-f).
\]
Over $\F$, the right side has a Boolean certificate. The same weight $f$
must be used throughout the derivation, so it must vanish on every proper
high-rank flat at once.

That is a linear-algebra problem. Let $\mathcal L_k$ be the space of
\emph{row-linear} polynomials of degree at most $k$: each monomial uses at
most one bit from each pigeon. Restriction to each flat is a linear map on
this space. If the sum of the image dimensions is smaller than
$\dim\mathcal L_k$, their common kernel contains a nonzero polynomial.
Lemma~\ref{lem:kernel} establishes this when
\[
 m\ln(4M)\le k^2,\qquad k\le m,
\]
where $M$ bounds the number of proper high-rank blocks. The loss of
dimension on one flat is strong enough to pay for all $M$ restrictions.

We can now substitute every block at once and multiply every substituted
line by this common $f$. Low-rank blocks use exact packing, inconsistent
blocks use constants, and proper high-rank blocks use the cofactors above.
The weighted images of their axioms follow from the old base alone. The
same is true of every later line by the substitution lemma, through degree
\[
 B=k(D+1).
\]
The final line $1$ becomes $f$. All extension variables have disappeared,
but the surviving polynomial need not be the constant one. Every intermediate
line, including those suppressed in the diagrams, obeys the same degree
ceiling $B$. The common-kernel argument chooses $f$; it need not be linear
or have a formula we could prescribe in advance.

\subsection{A polynomial the base cannot derive}

Why should deriving $f$ be impossible? A lower bound for deriving $1$
does not answer that question: nonzero polynomials can be perfectly valid
consequences of satisfiable axioms. The useful extra fact is that our $f$
belongs to $\mathcal L_k$.

Theorem~\ref{thm:row-separation} excludes every nonzero member of this space
from degree-$B$ consequences of the compact pigeonhole axioms, provided
$4(k-1)<n$ and $2B-1\le n$. The proof decodes bit labels into a unary
system, whose variables $X_{iz}$ say that pigeon $i$ occupies hole $z$.
Products supported on partial matchings survive its Boolean and exclusion
relations; products with a collision do not.

A linear functional annihilating the low-degree consequences is therefore
specified by its values on partial matchings. The row equations require
these \emph{matching moments} to satisfy marginal equations. Prescribed
moments can be extended to higher degree in the range $2B-1\le n$
(Theorem~\ref{thm:moments}). This is a cycle-filling problem in a
chessboard complex: the known moments prescribe a boundary, the marginal
equations make it a cycle, and vanishing homology supplies a filling.
The required homological result is proved in Appendix~\ref{sec:topology}.

To separate a particular $f$, choose a top-degree monomial. It names $t$
rows and one bit in each. In every selected row, choose two hole labels
differing in that bit, with all pairs disjoint. Summing over the resulting
binary cube cancels every monomial that misses a selected row and turns
the chosen coefficient into $1$. The same operation lowers degree by $t$
and leaves a smaller pigeonhole instance. A normalized matching functional
for that instance then detects $f$ while annihilating every degree-$B$
consequence of the original base.

Combining the cube with a normalized matching functional separates $f$
from every degree-$B$ consequence of the old base. The polynomial left by
removal is therefore precisely the kind the base cannot derive. The size
assumption must make both constructions possible at the same degree.

\subsection{The parameter balance}
The three steps meet in the inequalities of Theorem~\ref{thm:affine}:
\[
 (n+1)\ln(4M)\le k^2,\qquad 2k(D+1)-1\le n,\qquad 4(k-1)<n .
\]
The separator leaves a budget of roughly $n/2$ for the whole transformed
derivation. With $D=12\ell+1$, this allows $k$ up to about $n/(24\ell)$.
Larger $k$ gives the common kernel more room, but also makes removal more
expensive. A choice of order $n/\ell$ leaves room for both steps, and then
the kernel estimate tolerates $\ln S$ of order $n/\ell^2$.
Section~\ref{sec:main-proof} makes that balance exact. The symbols and
the conditions they must satisfy are collected below.
\begin{center}
\renewcommand{\arraystretch}{1.35}
\begin{tabular}{@{}>{\raggedright\arraybackslash}p{0.27\textwidth}>{\raggedright\arraybackslash}p{0.68\textwidth}@{}}
\toprule
\(\ell\) & bits per hole label\\\addlinespace[0.2em]
\(n=2^\ell\) & number of holes\\\addlinespace[0.2em]
\(m\) & number of pigeons (rows); $m=n+1$ in the main theorem\\\addlinespace[0.2em]
\(v=m\ell\) & number of bit variables\\\addlinespace[0.2em]
\(N\) & number of columns of a unary system (\(N=n\) in the application)\\\addlinespace[0.2em]
\(S,\ N_{\mathrm{reg}}=3S+\binom m2\) & proof-lines of the refutation; blocks (slots) of the registry\\\addlinespace[0.2em]
\(h\) & number of factors of a block product (its \emph{accuracy}); \(h=3\ell\)\\\addlinespace[0.2em]
\(s\) & number of inputs of a block (clause width after compression, at most \(v+2\))\\\addlinespace[0.2em]
\(r\) & rank of a block: dimension of the span of its inputs; threshold \(h(k+1)\)\\\addlinespace[0.2em]
\(D\) & degree of the simulating PC refutation; \(D=12\ell+1\)\\\addlinespace[0.2em]
\(k\) & degree bound for \(f\) and for the substituted coefficients\\\addlinespace[0.2em]
\(t\) & degree of a particular \(f\); number of rows in the separating cube; \(t\le k\)\\\addlinespace[0.2em]
\(M\) & bound on the number of proper high-rank blocks; \(M\le N_{\mathrm{reg}}\)\\\addlinespace[0.2em]
\(B=k(D+1)\) & degree at which \(f\) would be derived; needs \(2B-1\le n\)\\\addlinespace[0.2em]
\bottomrule
\end{tabular}
\end{center}

The conditions have distinct roles; the exact joint statement is
Theorem~\ref{thm:affine}.
\begin{center}
\renewcommand{\arraystretch}{1.4}
\begin{tabular}{@{}>{\raggedright\arraybackslash}p{0.36\textwidth}>{\raggedright\arraybackslash}p{0.59\textwidth}@{}}
\toprule
Condition & Where it is used\\\addlinespace[0.2em]
\midrule
\(m\ln(4M)\le k^2,\ k\le m\) & Common-kernel dimension estimate
(Lemma~\ref{lem:kernel}).\\\addlinespace[0.2em]
\(r>h(k+1),\ h=3\ell\) & Proper high-rank blocks have sufficiently
small restriction spaces; lower ranks use packing.\\\addlinespace[0.2em]
\(D\ge2h+1\) & Weighted companion and Boolean certificates fit the
removal ceiling (Lemma~\ref{lem:hybrid}).\\\addlinespace[0.2em]
\(B=k(D+1)\) & Output degree of weighted removal, used as the input
ceiling for the separator.\\\addlinespace[0.2em]
\(4(k-1)<n\) & Cube separation for every target degree \(t\le k\)
(Theorem~\ref{thm:row-separation}).\\\addlinespace[0.2em]
\(2B-1\le n\) & Matching-moment and residual old-system range in
that separator, with \(n=2^\ell\).\\\addlinespace[0.2em]
\bottomrule
\end{tabular}
\end{center}

\subsection{Why depth and regularity disappear}
Restriction-based lower bounds for regular or bounded-depth $\Res$
\cite{EGI25,BCD24,BC25,BI25,BCBI26,AI25,EI25} control what can accumulate
along a path through the proof. Regularity or a depth bound supplies that
control. A general DAG can have long paths and can reuse a line many times.

The simulation above assigns a polynomial value to each proof-line and
uses local identities to derive the next value. Its degree does not depend
on the length of a path. The remaining obstruction is algebraic: can the
extension variables be removed? Affine inputs in the original variables
permit exact packing at low rank and common-kernel elimination at high rank.
The corresponding extension blocks for $\mathrm{AC}^0[p]$-Frege proofs
are nested~\cite{BIKPRS}; one block's inputs can involve another block's
variables. The removal argument here does not apply to that situation.

\FloatBarrier

\section{Ordinary degree and polynomial calculus}\label{sec:algebra}

Throughout, the field is $K=\F$ and variable sets are finite. Two polynomials
can agree on every binary assignment and still be different polynomials:
$x^2-x$ is the simplest example. Degree measures the expressions in the
derivation, so reducing powers must be justified by the Boolean equations
at the degree where the reduction is used.

Degree bounds and accuracy parameters are integers.
Polynomials are ordinary commutative polynomials. In particular, powers are
not silently reduced using Booleanity. We use the nonnegative total-degree
convention \(\deg0=0\). The Boolean equations are
\[
 \Bool_x=\{x_j^2-x_j:j\}.
\]
Every displayed generator is an equation with right side zero.

For a finite polynomial system \(\mathcal F\), define its degree-\(B\)
Nullstellensatz (NS) span by
\begin{equation}\label{eq:ns-space}
 \Ideal_B(\mathcal F)
  =\Span_K\{qF:F\in\mathcal F,\ \deg(qF)\le B\}
   \subseteq\Pol_{\le B}.
\end{equation}
The degree constraint is imposed on each generator multiple. Cancellation
between larger-degree multiples does not make them legal in
\(\Ideal_B\).

Ordinary PC derives axioms of \(\mathcal F\), linear combinations of
previous lines, and \(x_jg\) from a previous line \(g\).
A degree-\(B\) derivation has no line of degree exceeding \(B\).
Let \(\Cons_B(\mathcal F)\) be the space of polynomials derivable in this
way. Refutation means deriving one. The space \(\Cons_B\) is closed under
linear combinations, but equality with \(\Ideal_B\) requires a proof.

A derived line can be used again without repeating the work that produced
it. In polynomial calculus, this distinction controls the degree cost of
multiplication. If an earlier derivation reached degree $B$ but ended in a
low-degree polynomial $g$, multiplying $g$ need not add its multiplier's
degree to every earlier line.

\begin{samepage}
\begin{lemma}[Multiplying a completed line]\label{lem:pc-product}\lean{claims/PolynomialCalculusReuse.lean}
If \(g\) has a degree-\(B\) PC derivation and \(q\) is any polynomial,
then \(qg\) has a PC derivation through
\(\max\{B,\deg q+\deg g\}\).
Moreover \(\Ideal_B(\mathcal F)\subseteq\Cons_B(\mathcal F)\), and
\(q\Ideal_B(\mathcal F)\subseteq\Ideal_{B+\deg q}(\mathcal F)\).
\end{lemma}
\end{samepage}
\begin{proof}
Expand \(q\) into monomials. Starting from the already derived line \(g\),
multiply by the variables of each monomial one at a time, and then take
the indicated linear combination. Each new line has degree at most
\(\deg g+\deg q\). The zero case is immediate.
Applying this construction to each legal axiom multiple in
\eqref{eq:ns-space}, then summing, proves the second assertion.
Multiplying each legal generator multiple by \(q\) proves the last assertion.
There is no need to multiply the earlier derivation of \(g\).
\end{proof}

\begin{samepage}
\begin{lemma}[Degree-controlled Boolean reduction]\label{lem:boolean}\lean{claims/BooleanReduction.lean}
For every polynomial \(R\) of degree at most \(d\), its squarefree
remainder \(\overline R\) satisfies
\[
 R-\overline R\in\Ideal_d(\Bool_x).
\]
If \(R\) vanishes at every Boolean point, then
\(R\in\Ideal_d(\Bool_x)\). In particular, if \(\deg q\le k\), then
\[
 q^2-q\in\Ideal_{2k}(\Bool_x).
\]
\end{lemma}
\end{samepage}
\begin{proof}
For \(a\ge2\),
\[
 x^a-x=(x^2-x)(1+x+\cdots+x^{a-2}).
\]
Apply this identity to the powers in each monomial of \(R\), one variable
at a time. Every generator multiple has degree at most the original
monomial degree, so the first assertion follows.

A squarefree polynomial that vanishes on \(\{0,1\}^v\) is zero.
For example, write it as \(A+x_vB\) with \(A,B\) independent of \(x_v\).
Evaluations at \(x_v=0,1\), followed by induction on \(v\), give \(A=B=0\).
This proves the second assertion. Finally, every value of \(q\) at a
Boolean point belongs to \(\F\), where its square equals itself; apply the
second assertion at degree \(2k\).
\end{proof}

The distinction will matter again when we restrict to an affine flat.
On $x=0$, the polynomial $y^2-y$ vanishes at both binary points, yet its
ordinary restriction is the nonzero polynomial $y^2-y$. Thus pointwise
vanishing cannot by itself give membership in the ideal $(x)$.
Lemma~\ref{lem:restriction-ideal} uses zero ordinary restriction to obtain
the bounded cofactors needed for extension removal.

\begin{samepage}
\begin{lemma}[Substitution and fixed-weight replay]\label{lem:substitution}
\lean{claims/PolynomialCalculusSubstitution.lean}
Let \(\phi\) substitute polynomials of degree at most
\(T\ge0\) for source variables. Then \(\deg\phi(p)\le T\deg p\), and
\(\phi\) carries degree-\(D\) PC derivations to degree-\(TD\) derivations
from the substituted axioms. More generally, fix a target polynomial \(w\)
and a target system \(\mathcal G\). If each source axiom \(F\) of degree
at most \(D\) has its weighted image \(w\phi(F)\) derivable from
\(\mathcal G\) through \(TD+\deg w\), then so does the weighted image
of every line of a degree-\(D\) source derivation.
\end{lemma}
\end{samepage}
\begin{proof}
Expand into monomials to obtain the degree inequality. Induct on primitive
PC steps. Zero, addition and scalar multiplication commute with the weighted
map. At a nonzero multiplication \(p\mapsto xp\), ordinary degree gives
\(\deg p+1\le D\). The completed weighted line has degree at most
\(\deg w+T\deg p\); multiplying it by \(\phi(x)\) costs at most
\(\deg w+T(\deg p+1)\le TD+\deg w\), by
Lemma~\ref{lem:pc-product}. A zero predecessor stays zero. Taking \(w=1\)
and the substituted axioms as target system gives ordinary substitution.
\end{proof}

A linear functional can certify that a desired polynomial lies outside the
space of low-degree consequences. It assigns zero to every consequence and
one to the target. The extension statement below lets us prescribe that
functional first on a smaller space and then extend it consistently.

\begin{samepage}
\begin{lemma}[Separation and compatible extension]\label{lem:duality}
\leangroup{\leanref{third-party-claims/LinearSeparation.lean}{extension}
  \leansep \leanref{claims/PolynomialCalculusDuality.lean}{PC duality}}
Let \(U,S\) be subspaces of a vector space \(V\) over \(K\).
A linear functional on \(U\) extends to one on \(V\) vanishing on \(S\)
if and only if it vanishes on \(U\cap S\).
In particular, if \(p\notin S\), there is a linear functional
\(\lambda\) with \(\lambda(S)=0\) and \(\lambda(p)=1\).
Consequently, for either \(S=\Ideal_B(\mathcal F)\) or
\(S=\Cons_B(\mathcal F)\), a normalized annihilator on
\(\Pol_{\le B}\) exists exactly when \(1\notin S\).
\end{lemma}
\end{samepage}
\begin{proof}
Define the extension on \(U+S\) by \(\lambda(u+s)=\lambda(u)\).
The intersection condition makes this well-defined. Extend a basis of
\(U+S\) to a basis of \(V\), and assign zero on the added basis vectors.
For separation, start on \(Kp\) with \(\lambda(cp)=c\); its intersection
with \(S\) is zero. Necessity follows by restricting a vanishing functional.
The polynomial consequences follow by applying this construction inside
\(\Pol_{\le B}\).
\end{proof}

\section{Matching moments and low-degree consequences}\label{sec:moments}

The pigeonhole axioms have no solution when there are too many pigeons.
Why should low-degree reasoning fail to detect this? A low-degree monomial
can describe only a partial assignment. After Boolean reduction and the
exclusion of collisions, its surviving support is a partial matching.
There can be plenty of these even when a full matching is impossible.
The housing shortage becomes apparent only when enough pigeons turn up.

A functional annihilating the low-degree consequences assigns compatible
values to those partial matchings. If its value on $1$ is one, the calculus
cannot derive $1$ at that degree. The issue is whether all the compatibility
conditions can be satisfied together.

For \(m\ge0\) and \(N\ge1\), the functional unary system
\(\mathcal F^{\mathrm{fun}}_{m,N}\) has variables \(X_{ij}\) and generators
\begin{align*}
 &X_{ij}^2-X_{ij},\\
 &X_{ij}X_{ij'} &&(j\ne j'),\\
 &X_{ij}X_{i'j} &&(i\ne i'),\\
 &\rho_i-1,\qquad \rho_i=\sum_{j=1}^N X_{ij}.&
\end{align*}
The row exclusions in the second line are part of this auxiliary system.
They are not inferred from a weaker unary encoding.

A matching \(T\) is a set of cells with distinct rows and distinct columns;
write \(X_T=\prod_{(i,j)\in T}X_{ij}\).
The empty matching corresponds to the constant one.

\begin{samepage}
\begin{lemma}[Complete moment equations]\label{lem:marginals}\leangroup{\leanref{claims/MatchingNormalForm.lean}{normal form}
  \leansep \leanref{claims/MatchingMomentCompleteness.lean}{moment equations}}
A linear functional on \(\Pol_{\le B}\) annihilating
\(\Ideal_B(\mathcal F^{\mathrm{fun}}_{m,N})\) is exactly a choice of
matching moments \(z_T\), \(|T|\le B\), satisfying
\begin{equation}\label{eq:marginals}
 \sum_{j\notin\cols(T)}z_{T\cup\{(i,j)\}}=z_T
 \quad\text{for }|T|<B,\ i\notin\rows(T).
\end{equation}
Its value on an ordinary monomial is the moment of its squarefree support
if that support is a matching, and is zero otherwise. Normalization is the
additional condition \(z_\varnothing=1\).
\end{lemma}
\end{samepage}
\begin{proof}
Boolean equations reduce powers within the original degree.
A squarefree monomial whose support has a repeated row or column contains
one of the exclusion generators as a factor. Thus every polynomial reduces,
within its degree, to a linear combination of matching monomials.
Conversely, the stated rule for evaluating ordinary monomials annihilates
every legal Boolean or exclusion multiple.

It remains to check row equations. In a row-equation multiple
\(q(\rho_i-1)\) of degree at most \(B\), we may first put \(q\) into the
matching normal form: any error is a Boolean or exclusion consequence
through \(B-1\), and multiplication by \(\rho_i-1\) keeps its certificate
within \(B\). For a matching multiplier \(X_T\), if row \(i\) is already
occupied, the product reduces to zero. Otherwise its reduction is
\[
 \sum_{j\notin\cols(T)}X_{T\cup\{(i,j)\}}-X_T.
 \]
These are precisely \eqref{eq:marginals}, and they exhaust the constraints.
No positivity or multiplicativity is imposed on the functional.
\end{proof}

The equations remove one cell from a matching. For two rows, put a
vertex at each cell and join two vertices when the cells can be occupied
together. An unknown edge moment contributes to the equations at both
endpoints. Giving coefficient one to a path leaves its endpoints:
every internal vertex occurs twice and cancels over $\F$.

Choose four distinct hole labels $\alpha,\beta,\gamma,\delta$, and consider
\[
 A=(2,\alpha),\quad B=(3,\beta),\quad K=(2,\gamma),\quad
 E=(3,\alpha),\quad U=(4,\delta).
\]
Prescribe singleton moment one at $A,E,U$ and zero at the other cells of
these rows. There is no matching edge $AE$, since the cells share a column.
The three-edge path $AB+BK+KE$ nevertheless gives the required moments
on rows 2 and 3: the contributions at $B$ and $K$ cancel. For pairs
involving row 4, use the edges $AU$ and $EU$.

What happens when we add triple moments? Three compatible cells span a
triangle. Its coefficient contributes to each of its three edges, so the
pair moments prescribe the boundary of a sum of triangles. Here that
boundary is the five-edge cycle
\[
 AB+BK+KE+EU+UA.
\]
The triangles $UAB$, $UBK$, and $UKE$ fill it. Their internal edges $UB$
and $UK$ each occur twice, while the five prescribed edges remain
(Figure~\ref{fig:moment-filling}). For example, the marginal at $UB$ is
$1+1=0$, while the marginal at $UA$ is one.

\momentfigure

These moments are linear data over $\F$, not probabilities. In particular,
singleton moments one at both $A$ and $E$ do not describe a random matching
that puts two pigeons in column $\alpha$. The marginal equations, rather
than positivity, govern which data extend.

The vertices, matching edges, triangles, and higher matching faces form
the \emph{chessboard complex} $\Delta_{s,N}$ on an $s$-row, $N$-column
board. The unknown moments on a fixed set of $s$ rows are coefficients
of an $(s-1)$-chain. Their prescribed boundary is a cycle because the
lower moments already satisfy the marginal equations.

Our three triangles give an explicit filling on four columns. This
particular cycle fills even below the uniform three-row guarantee
$N\ge5$. Its triangles have the common vertex $U$; a general boundary need
not have such a convenient filling; a cycle alone is not enough. The
six-cycle on a two-row, three-column board, for instance, has no triangles
at all. The needed guarantee comes from the chessboard connectivity theorem
of Bj\"orner, Lov\'asz, Vre\'cica, and \v{Z}ivaljevi\'c
\cite{BLVZ94}. Its following homological consequence, with the required
chain arguments, is proved in Appendix~\ref{sec:topology}:
\begin{equation}\label{eq:chess-vanish}
 s\ge2,\quad N\ge2s-1
 \quad\Longrightarrow\quad
 \widetilde H_{s-2}(\Delta_{s,N};\F)=0.
\end{equation}
Thus every compatible boundary in this range has a filling. Different row
sets introduce different new moments, so their fillings agree on the
already prescribed lower-degree data.

The resulting degree range is Razborov's~\cite{Raz98}; his analysis also
describes the low-degree consequences through an explicit basis. The
argument here gives a self-contained proof over $\F$ in the form of
extension of every prescribed functional. No new degree range is claimed.

\begin{samepage}
\begin{theorem}[Extension of prescribed matching moments]\label{thm:moments}\leangroup{\leanref{claims/MatchingMomentExtension.lean}{extension}
  \leansep \leanref{claims/MatchingFiltration.lean}{annihilators}}
Suppose \(m\ge0\), \(N\ge\max\{1,2B-1\}\), and \(0\le k\le B\).
Every functional annihilating
\(\Ideal_k(\mathcal F^{\mathrm{fun}}_{m,N})\) extends to one annihilating
\(\Ideal_B(\mathcal F^{\mathrm{fun}}_{m,N})\).
The constant moment is preserved. In particular normalized functionals
exist through degree \(B\).
\end{theorem}
\end{samepage}
\begin{proof}
By Lemma~\ref{lem:marginals}, it suffices to extend the matching moments.
Starting from the prescribed values, we fill degrees \(s=k+1,\ldots,B\).
At degree one choose moments in each row whose sum is the prescribed
constant moment; a column exists because \(N\ge1\).

For \(s\ge2\), fix a set \(S\) of \(s\) rows. Form the reduced
\((s-2)\)-chain
\[
 z_S=\sum_{\substack{T\text{ a matching on }S\times[N]\\|T|=s-1}}
             z_T[T]
\]
in \(\Delta_{s,N}\).
This is a cycle. At \(s=2\), its augmentation is the sum of two equal
row totals, which is zero in \(\F\).
For \(s>2\), a face \(U\) of size \(s-2\) misses two rows of \(S\).
Its boundary coefficient is the sum of the prescribed marginal sums for
those two rows, hence \(z_U+z_U=0\).

Since \(N\ge2s-1\), equation~\eqref{eq:chess-vanish} gives an
\((s-1)\)-chain \(y_S\) with boundary \(z_S\).
Assign the coefficients of \(y_S\) as the new moments on full
\(s\)-matchings of \(S\times[N]\).
For a fixed \((s-1)\)-matching, its boundary equation is exactly the
missing-row equation in \eqref{eq:marginals}.

Different \(s\)-row sets have disjoint new unknowns: a size-\(s\) matching
has a unique occupied row set. Their shared lower faces retain their
already assigned values. Thus these fillings are globally consistent and
complete the induction. If \(s>m\), there are no new row sets or
moments to fill. In particular, no hypothesis \(m\ge B\) is required.
The argument works for either value of the constant moment.
Starting with \(z_\varnothing=1\) at degree zero supplies a normalized
functional.
\end{proof}

This extension property prevents a high-degree certificate from creating
a genuinely new low-degree consequence within the stated range. If a
low-degree polynomial were new, a separating functional at that degree
would extend to the larger degree and still detect it. The same fact
controls multiplication steps in polynomial calculus.

\begin{samepage}
\begin{corollary}[Old filtration stability and PC closure]\label{cor:stability}\lean{claims/MatchingFiltration.lean}
Under \(m\ge0\) and \(N\ge\max\{1,2B-1\}\),
\begin{equation}\label{eq:old-stability}
 \Ideal_B(\mathcal F^{\mathrm{fun}}_{m,N})\cap\Pol_{\le k}
       =\Ideal_k(\mathcal F^{\mathrm{fun}}_{m,N})
       \quad(0\le k\le B),
 \qquad
 \Cons_B(\mathcal F^{\mathrm{fun}}_{m,N})
       =\Ideal_B(\mathcal F^{\mathrm{fun}}_{m,N}).
\end{equation}
In particular \(1\notin\Cons_B(\mathcal F^{\mathrm{fun}}_{m,N})\).
\end{corollary}
\end{samepage}
\begin{proof}
If a polynomial of degree at most \(k\) is outside \(\Ideal_k\),
finite-dimensional linear duality gives a functional vanishing on
\(\Ideal_k\) but not on that polynomial. Extend it by
Theorem~\ref{thm:moments}. The polynomial is then outside \(\Ideal_B\)
as well. The reverse inclusion in the first equality is immediate.

To prove PC closure, induct over a degree-\(B\) derivation.
Axioms and linear combinations lie in \(\Ideal_B\).
At a nonzero multiplication step \(g\mapsto X_{ij}g\), the predecessor
has ordinary degree at most \(B-1\). If \(g\in\Ideal_B\), the first
equality puts it in \(\Ideal_{B-1}\). Multiplication of its legal
generator multiples by \(X_{ij}\) stays in \(\Ideal_B\).
Hence \(\Cons_B\subseteq\Ideal_B\); the other direction is
Lemma~\ref{lem:pc-product}.
A normalized annihilator excludes one.
\end{proof}

\section{The bit decoder and a separating functional}\label{sec:decoder}

The matching argument uses one variable for each pigeon--hole pair.
Our formula uses bits instead. We need a translation that preserves both
the degree of a derivation and the nonzero polynomial that it produces.
Here $n=2^\ell$ with $\ell\ge2$, and $m$ is arbitrary.
The compact old base is the system defined in \eqref{eq:compact-base}.
The degree of each collision generator is \(\ell\).
Its vanishing asserts that the two bit rows encode different labels.

For two-bit rows $x=(x_1,x_2)$ and $y=(y_1,y_2)$, the collision polynomial is
\[
 E(x,y)=(1-x_1-y_1)(1-x_2-y_2).
\]
On binary assignments, it equals one when the two labels coincide and zero
when either bit differs. Thus the single equation $E=0$ excludes every
collision between those rows, rather than naming one hole at a time.

Identify unary columns with labels \(z\in\{0,1\}^{\ell}\), and define the
ordinary ring map
\begin{equation}\label{eq:decoder}
 \tau(b_{it})=\sum_{z:z_t=1}X_{iz}.
\end{equation}

\begin{samepage}
\begin{lemma}[Two-row interpolation]\label{lem:interpolation}\lean{claims/TwoRowInterpolation.lean}
Let \(\mathcal U_{i,i'}\) consist of Booleanity, same-row exclusions,
and the two row-sum equations for distinct rows \(i,i'\).
For a polynomial \(R\) supported on these rows and of degree at most \(d\),
\[
 R-\sum_{z,w}R(e_z,e_w)X_{iz}X_{i'w}
     \in\Ideal_{\max\{d,2\}}(\mathcal U_{i,i'}).
\]
Here \(e_z,e_w\) are the corresponding one-hot assignments.
Column exclusions are not used in this interpolation.
\end{lemma}
\end{samepage}
\begin{proof}
Booleanity and the same-row exclusions reduce powers and collisions.
The remainder is a linear
combination of \(1\), individual cells of the two rows, and products
containing one cell from each row. This reduction costs at most \(d\).
Replace the constant by \(\rho_i\rho_{i'}\), using
\[
 1-\rho_i\rho_{i'}=(1-\rho_i)+\rho_i(1-\rho_{i'}),
\]
and replace each remaining single cell by its product with the other
row sum. These changes have certificates of degree at most two.
The result is bilinear, and its coefficient at \(X_{iz}X_{i'w}\)
is its value at \((e_z,e_w)\), namely \(R(e_z,e_w)\).
\end{proof}

The decoder is linear on variables, so it cannot increase degree. Could
it lower the degree or cancel a nonzero polynomial? The unit bit labels
supply a polynomial left inverse, which rules out both possibilities.

\begin{samepage}
\begin{lemma}[A polynomial left inverse for the decoder]\label{lem:decoder-injective}
\lean{claims/CompactBitDecoder.lean}
For every \(m,\ell\ge0\), the decoder \(\tau\) is injective and preserves
ordinary total degree.
\end{lemma}
\end{samepage}
\begin{proof}
Define a linear substitution \(\lambda\) on unary variables by sending
\(X_{i,e_t}\) to \(b_{it}\), where \(e_t\) is the unit bit label, and
all other unary variables to zero. The unit labels are distinct and
\(\lambda\tau(b_{it})=b_{it}\), so \(\lambda\tau\) is the identity
on polynomials. Both substitutions have degree at most one; applying the
degree inequality in both directions proves equality, as well as injectivity.
\end{proof}

\begin{samepage}
\begin{lemma}[Degree-preserving old-base transfer]\label{lem:decoder}
\lean{claims/CompactBitDecoderTransfer.lean}
For \(m\ge0\), \(\ell\ge2\), and every \(B\ge0\),
\[
 q\in\Cons_B(\mathcal Q_{m,\ell})
 \quad\Longrightarrow\quad
 \tau q\in\Cons_B(\mathcal F^{\mathrm{fun}}_{m,n}).
\]
\end{lemma}
\end{samepage}
\begin{proof}
Bit Booleanity has a degree-two image certificate:
\[
 \tau(b_{it}^2-b_{it})=\sum_{z:z_t=1}(X_{iz}^2-X_{iz}).
\]
At a two-row one-hot assignment \((e_z,e_w)\), \(\tau(E_{ii'})\)
is one exactly when \(z=w\). Lemma~\ref{lem:interpolation} therefore gives
\[
 \tau(E_{ii'})-\sum_zX_{iz}X_{i'z}
       \in\Ideal_\ell(\mathcal U_{i,i'}).
\]
The remaining terms are column exclusions of degree two. Thus every axiom
image has a certificate through its original degree. Replay the source
proof using Lemma~\ref{lem:substitution} with \(T=1\), \(w=1\).
Only axioms admitted by the source degree ceiling are replayed, so
\(B\ge\ell\) need not be assumed.
\end{proof}

Let \(\mathcal L_k\) be the space spanned by monomials of degree at most
\(k\) containing at most one bit variable from each row.

\begin{samepage}
\begin{lemma}[Row-linear dimension]\label{lem:row-dimension}
\lean{claims/RowLinearPolynomialSpace.lean}
For all \(m,\ell,k\ge0\),
\begin{equation}\label{eq:row-linear-dimension}
 \dim\mathcal L_k=\sum_{j=0}^k\binom mj\ell^j.
\end{equation}
Every nonzero polynomial in this space has a nonzero monomial coefficient
of degree equal to its total degree.
\end{lemma}
\end{samepage}
\begin{proof}
Ordinary monomials form a basis. To choose a row-linear monomial of degree
\(j\), choose its \(j\) rows and one of \(\ell\) coordinates in each.
Terms with \(j>m\) contribute zero. The last assertion follows by choosing
a maximal-degree element of the finite nonempty support.
\end{proof}

A top-degree monomial can be isolated by varying one selected bit in
each of its rows. For instance, take three-bit rows $x,y$ and
$f=x_1y_2+x_2+1$. Let $x$ range over $000,100$ and $y$ over $001,011$.
The four labels are distinct, and the four evaluations are
\[
\begin{array}{c|cc}
 &y=001&y=011\\\hline
 x=000&1&1\\
 x=100&1&0
\end{array}
\]
Their sum in $\F$ is one: the coefficient of $x_1y_2$. The constant term
occurs four times, and a term depending on only one of the two varied bits
occurs in equal pairs. The same cancellation works for an arbitrary
row-linear polynomial, with a degree drop for the remaining terms.

\begin{samepage}
\begin{lemma}[Binary cube degree drop]\label{lem:cube-drop}
\lean{claims/BinaryCubeDegreeDrop.lean}
Partition some source variables into \(t\) indexed groups. For each
\(\epsilon\in\F^t\), substitute scalars depending only on \(\epsilon_i\)
for variables in group \(i\), and fixed affine target polynomials for
all remaining variables. Write \(\phi_\epsilon\) for this substitution
and \(\Delta p=\sum_\epsilon\phi_\epsilon(p)\). Then
\[
 \deg\Delta p\le\max\{\deg p-t,0\},\qquad
 \deg p<t\Longrightarrow\Delta p=0.
\]
If every source generator \(F\) has a cube-independent image \(r_F\)
which is zero or a target generator, then
\[
 \Delta\Ideal_B(\mathcal F)\subseteq
       \Ideal_{\max\{B-t,0\}}(\mathcal G).
\]
\end{lemma}
\end{samepage}
\begin{proof}
A monomial missing one group has equal evaluations in pairs obtained by
flipping that cube coordinate, so its sum is zero. A surviving monomial
uses at least \(t\) selected variables, which become scalars; the remaining
factors have total degree at most its original degree minus \(t\).
For a legal multiple \(qF\), cube independence gives
\(\Delta(qF)=(\Delta q)r_F\). Apply the degree-drop bound to the whole
source polynomial \(qF\), whose degree is at most \(B\). Thus this is
one target generator multiple of degree at most \(B-t\) when \(B\ge t\),
or zero when \(B<t\); it also vanishes if \(r_F=0\). Extend by linearity.
\end{proof}

\begin{samepage}
\begin{lemma}[Row-cube restriction and coefficient isolation]\label{lem:row-cube}
\leangroup{\leanref{claims/DisjointCoordinatePairs.lean}{pair packing}
  \leansep \leanref{claims/RowCubeRestriction.lean}{restriction}
  \leansep \leanref{claims/RowCubeCoefficientIsolation.lean}{coefficient isolation}}
Choose \(t\) distinct rows and, for each, two labels differing only in
one prescribed bit. If these pairs are disjoint, the one-hot substitutions
on the selected rows, together with deletion of their \(2t\) columns on
other rows, give
\[
 \Delta\Ideal_B(\mathcal F^{\mathrm{fun}}_{m,n})
 \subseteq\Ideal_{\max\{B-t,0\}}
       (\mathcal F^{\mathrm{fun}}_{m-t,n-2t}).
\]
For a row-linear \(f\) of degree at most \(t\), this cube applied to
\(\tau f\) is the coefficient of the monomial selecting those
rows and prescribed bits. Such disjoint pairs can be chosen whenever
\(t\ge1\) and \(4(t-1)<2^\ell\).
More generally, the coefficient-isolation identity holds for arbitrary
cube-independent polynomial images of the unselected bits, provided each
selected bit has its prescribed base value plus its selected cube coordinate,
and the other bits in that row stay constant.
\end{lemma}
\end{samepage}
\begin{proof}
On a selected row, each substitution is a one-hot assignment; on every
other row, deleted columns become zero and surviving variables are relabeled.
Booleanity, row exclusions and column exclusions consequently map to zero
or the corresponding residual generator. Disjointness prevents a collision
between selected rows or between selected and residual rows. A selected
row-sum equation maps to zero, and every other row-sum equation maps to
its residual counterpart. Lemma~\ref{lem:cube-drop} applies.

A row-linear monomial missing a selected row cancels by a cube flip.
Each such monomial is counted twice---a fatal clerical error in
characteristic two. If it
uses all \(t\) selected rows, the degree bound forces it to use exactly
one variable in each and none outside. The two-value sum in each row is
one for the prescribed bit and zero for every other bit. This isolates
exactly the claimed coefficient, independently of the images outside the
selected rows. Finally, after selecting \(j\) pairs, their \(2j\) vertices
forbid at most \(2j\) of the \(2^{\ell-1}\) pairs in any specified direction.
The strict packing bound permits the next choice.
\end{proof}

The extension-removal step will produce a nonzero row-linear $f$, not
necessarily the constant one. Can the same obstruction detect every such
$f$? Its top-degree coefficient supplies a choice of cube, and that cube
turns $f$ into the constant one on a smaller instance. Pick a
monomial of top degree \(t\) in \(f\); it names \(t\) rows and one bit in
each. For each of these rows choose two labels that differ exactly in the
named bit, and sum over the \(2^t\) ways of assigning one of its two labels
to each selected row. Over \(\F\) this sum \(\Delta\) annihilates
every monomial that misses a selected row, so it sends \(\tau f\) to the
constant one (the chosen coefficient), lowers every degree by \(t\), and
sends pigeonhole consequences to consequences of the pigeonhole principle
on the remaining \(m-t\) rows and \(n-2t\) columns. A normalized functional
\(\lambda\) for that smaller instance, which exists by
Theorem~\ref{thm:moments}, then gives \(\mu=\lambda\circ\Delta\).

\begin{samepage}
\begin{theorem}[Row-linear separation through a larger PC degree]
\label{thm:row-separation}\lean{claims/CubeResidualDualSeparation.lean}
Let \(m\ge0\), \(\ell\ge2\), \(n=2^\ell\), and suppose
\(1\le k\le B\), \(4(k-1)<n\), and \(2B-1\le n\).
For every nonzero \(f\in\mathcal L_k\), there is a linear functional
\(\mu\) on unary polynomials through degree \(B\) such that
\[
 \mu(\Cons_B(\mathcal F^{\mathrm{fun}}_{m,n}))=0,
 \qquad \mu(\tau f)=1.
\]
It can be chosen with \(\mu(1)=0\) if \(\deg f>0\), and
\(\mu(1)=1\) if \(\deg f=0\). Consequently
\[
 \tau f\notin\Cons_B(\mathcal F^{\mathrm{fun}}_{m,n}),\qquad
 f\notin\Cons_B(\mathcal Q_{m,\ell}).
\]
\end{theorem}
\end{samepage}
\begin{proof}
A nonzero constant over \(\F\) is one, so that case follows from a
normalized matching functional. Otherwise put \(t=\deg f\) and choose a
nonzero degree-\(t\) coefficient. Lemma~\ref{lem:row-cube} supplies its
coordinate pairs and residual operator \(\Delta\), with \(\Delta\tau f=1\).
Because \(n\) is divisible by four, \(4(k-1)<n\) implies \(k\le n/4\).
Thus \(n-2t\ge1\), and
\[
 n-2t\ge2(B-t)-1.
\]
Theorem~\ref{thm:moments} supplies a normalized residual annihilator
\(\lambda\) through \(B-t\), even if there are fewer than \(B-t\)
remaining rows. Set \(\mu=\lambda\circ\Delta\).
To see that \(\mu\) annihilates \(\Cons_B(\mathcal F^{\mathrm{fun}}_{m,n})\),
note first that this space equals \(\Ideal_B(\mathcal F^{\mathrm{fun}}_{m,n})\)
by Corollary~\ref{cor:stability}, which applies since \(2B-1\le n\).
The first assertion of Lemma~\ref{lem:row-cube} gives
\(\Delta\Ideal_B(\mathcal F^{\mathrm{fun}}_{m,n})\subseteq
\Ideal_{B-t}(\mathcal F^{\mathrm{fun}}_{m-t,n-2t})\), and \(\lambda\)
vanishes on the latter space. Coefficient isolation gives
\(\mu(\tau f)=1\), whereas \(\Delta1=0\) gives \(\mu(1)=0\).
Finally apply Lemma~\ref{lem:decoder} for the bit consequence.
\end{proof}

\paragraph{The local parity and its separator.}\label{sec:local-separator}
For the four-bit rows of Section~\ref{sec:local-resolution}, take
$f_*=x+y$. Set pigeon 1 to $0000$ and $1000$, delete these two columns
from the other rows in both restrictions, and sum. Then
\[
 \Delta\tau f_*=(0+\tau y)+(1+\tau y)=1.
\]
The residual board has sixteen rows and fourteen columns. The filling
in Figure~\ref{fig:moment-filling} applies with
\[
 \alpha=0001,\qquad\beta=0010,\qquad
 \gamma=0011,\qquad\delta=0101.
\]
Give every residual row except row 4 singleton moment one at $\alpha$,
and row 4 singleton moment one at $\delta$. For each pair of rows other than row 4, use the same three-edge path through $\alpha,\beta,\gamma$;
for a pair involving row 4, use the single edge between $\alpha$ and
$\delta$. These choices satisfy all pair marginals. On rows 2, 3, and 4,
retain the three illustrated triangles. The other triples have compatible
boundary cycles and hence fillings by \eqref{eq:chess-vanish}.
Each triple adds only its own new moments, so these choices agree.

Theorem~\ref{thm:moments} extends the resulting normalized functional
through degree six, since $14\ge2\cdot6-1$. With this extension $\lambda$,
\[
 \Lambda_*:=\lambda\circ\Delta\circ\tau,\qquad
 \Lambda_*(f_*)=1,\qquad
 \Lambda_*(\Cons_7(\mathcal Q_{17,4}))=0.
\]
The separator's conditions are $2\cdot7-1\le16$ and $4(1-1)<16$.
Thus the same parity used in the local resolution calculation cannot be
derived from the old bit-PHP base through degree seven. For a general
nonzero $f$, the cube isolates a top-degree coefficient in exactly the
same way; the common-kernel construction is free to choose that $f$.

\section{A direct PC simulation of the proof DAG}\label{sec:clause-simulation}

To simulate a proof DAG, each proof-line must receive a polynomial value,
and every inference must derive its conclusion's value without raising
the global degree ceiling. The block family is fixed before the derivation
begins, including any auxiliary clauses needed to express a binary rule.

\begin{samepage}
The product construction belongs to the extension-Nullstellensatz approach
of Buss, Impagliazzo, Kraj\'i\v{c}ek, Pudl\'ak, Razborov, and Sgall
\cite{BIKPRS}. We use a complete one-level family and prove its required
properties explicitly. An accuracy-\(h\) block has a finite tuple of old
affine polynomials \(g_1,\ldots,g_s\), fresh coefficient variables
\(r_{ui}\), and product
\begin{equation}\label{eq:block}
 P=\prod_{u=1}^h\left(1-\sum_{i=1}^s r_{ui}g_i\right).
\end{equation}
Its complete axioms are all companions \(g_iP\) and all coefficient
Boolean equations \(r_{ui}^2-r_{ui}\). The product \(P\) itself is not an
axiom. Coefficients of different blocks are disjoint; every input depends
only on the old variables. Rank means the dimension of the affine input
span. A nonzero span is called proper if it does not contain one.

\end{samepage}

Unless stated otherwise, \(h\ge1\) and there are finitely many old
variables, say \(v\). Represent a linear clause by its affine true-indicators:
\[
 C=\bigvee_i(g_{C,i}=1),\qquad
 Z(C)=\{b:g_{C,i}(b)=0\text{ for all }i\}.
\]
An equation \(a(b)=\alpha\) has true-indicator \(1+a(b)+\alpha\).
The set \(Z(C)\) is an affine subspace or empty.

Assign to each nonempty clause a fresh complete accuracy-\(h\) block,
and write its product as \(P_C\). Set \(P_\varnothing=1\), using no
block for the empty clause. \lean{claims/AffineClauseRegistry.lean} Fix the whole block family before following
the proof DAG. Clause values \(P_C=0\) will be derived; they are not
included among its axioms.

\sourceonlyfigure

\begin{samepage}
\begin{lemma}[General product telescoping]\label{lem:telescoping}
\lean{claims/MpTelescoping.lean}
For polynomials \(g_i\), \(r_{ui}\) with finitely many indices \(i\), and \(h\ge0\), put
\[
 P_j=\prod_{u=0}^{j-1}\left(1-\sum_i r_{ui}g_i\right),\qquad
 U_{h,i}=\sum_{j=0}^{h-1}r_{ji}P_j.
\]
Then \(1-P_h=\sum_i U_{h,i}g_i\).
If \(\deg g_i\le\delta\), \(\deg r_{ji}\le1\), and \(h\ge1\), then
\[
 \deg P_h\le h(\delta+1),\qquad
 \deg U_{h,i}\le1+(h-1)(\delta+1),\qquad
 \deg(g_iP_h)\le\deg g_i+h(\delta+1).
\]
These are upper bounds; no freshness or nonvanishing of the inputs is assumed.
\end{lemma}
\end{samepage}
\begin{proof}
The identity holds at \(h=0\). Writing \(s_h=\sum_i r_{hi}g_i\), the
recurrences \(P_{h+1}=P_h(1-s_h)\) and
\(U_{h+1,i}=U_{h,i}+r_{hi}P_h\) give
\[
 1-P_{h+1}=1-P_h+P_hs_h=\sum_i U_{h+1,i}g_i.
\]
Each factor has degree at most \(\delta+1\), and each summand of
\(U_{h,i}\) has degree at most \(1+(h-1)(\delta+1)\).
The sum and product degree inequalities prove the three bounds. At \(h=0\), \(P_0=1\) and \(U_{0,i}=0\).
\end{proof}

\begin{samepage}
\begin{corollary}[Affine-clause prefix identity]\label{lem:prefix}
\lean{claims/MpTelescoping.lean}
Every accuracy-\(h\) clause block, \(h\ge1\), has coefficients satisfying
\begin{equation}\label{eq:prefix}
 1-P_C=\sum_iU_{C,i}g_{C,i},\qquad
 \deg P_C\le2h,\qquad\deg U_{C,i}\le2h-1.
\end{equation}
Its companions have degree at most \(2h+1\).
\end{corollary}
\end{samepage}
\begin{proof}
Apply Lemma~\ref{lem:telescoping} with \(\delta=1\), reindexing coefficient
rows from zero-based to one-based. Exact companion degrees require the
additional hypotheses in Lemma~\ref{lem:cleanup}; a loose input-degree
bound alone would not justify equality.
\end{proof}

For a resolution step, the pivot should cancel while the side literals
are absorbed by companions of the conclusion. The prefix identity performs
exactly this separation. Its coefficients have degree at most $2h-1$,
which keeps the correction for the affine pivot within $4h+1$.

\begin{samepage}
\begin{lemma}[Complementary-parity resolution]\label{lem:clause-resolution}\lean{claims/AffineClauseResolution.lean}
Suppose \(A=(\bigvee_i(a_i=1))\vee(u=1)\) and
\(B=(\bigvee_j(b_j=1))\vee(1-u=1)\), with \(u\) affine.
Let \(T\) have the union of the context inputs \(a_i,b_j\).
If \(P_A,P_B\) have PC derivations through degree \(K\), then \(P_T\)
has a derivation in the same fixed block system through
\(\max\{K,4h+1\}\).
\end{lemma}
\end{samepage}
\begin{proof}
The context companions \(a_iP_T,b_jP_T\) are axioms of the conclusion
block; a repeated input may use the same axiom.
The prefix identities yield
\begin{align}
 R_A&=P_T-U_{A,u}uP_T
       =P_AP_T+\sum_iU_{A,i}(a_iP_T),\label{eq:resolution-A}\\
 R_B&=P_T-U_{B,1-u}(1-u)P_T
       =P_BP_T+\sum_jU_{B,j}(b_jP_T).\label{eq:resolution-B}
\end{align}
Lemma~\ref{lem:pc-product} derives these through \(\max\{K,4h\}\).
In particular it multiplies the completed lines \(P_A,P_B\), not their
entire earlier derivations.

For an affine \(u=c+\sum_t c_tb_t\), the identity
\[
 u^2-u=\sum_t c_t(b_t^2-b_t)
\]
has a degree-two proof. Therefore
\begin{equation}\label{eq:pivot-correction}
 (1-u)R_A-U_{A,u}P_T(u^2-u)=(1-u)P_T
\end{equation}
is derivable through \(\max\{K,4h+1\}\).
The right side is a completed line of degree at most \(2h+1\).
Multiply it by \(U_{B,1-u}\), costing at most \(4h\), and add
\eqref{eq:resolution-B}. The result is \(P_T\).
All multiplications expand into primitive PC steps by
Lemma~\ref{lem:pc-product}.
\end{proof}

\begin{samepage}
\begin{lemma}[Semantic weakening and tautologies]\label{lem:clause-weakening}\lean{claims/AffineClauseWeakening.lean}
If \(C\models D\) and \(P_C\) has a degree-\(K\) derivation, then
\(P_D\) is derivable through \(\max\{K,4h\}\).
A tautological \(D\) has a direct value derivation through \(2h+1\).
\end{lemma}
\end{samepage}
\begin{proof}
If \(Z(D)\) is nonempty, then \(Z(D)\subseteq Z(C)\).
Each \(g_{C,i}\) vanishes on \(Z(D)\), so affine linear algebra gives
constants \(\lambda_{ij}\) with
\[
 g_{C,i}=\sum_j\lambda_{ij}g_{D,j}.
\]
To see this directly, choose affine coordinates beginning with a basis
of the equations defining \(Z(D)\). An affine polynomial vanishing when
these coordinates are zero has only those coordinate terms.
Consequently \(g_{C,i}P_D\) is a constant linear combination of
conclusion companions. Now use
\[
 P_D=P_CP_D+\sum_iU_{C,i}(g_{C,i}P_D)
\]
and Lemma~\ref{lem:pc-product}, with ceiling \(\max\{K,4h\}\).

If \(Z(D)\) is empty, its affine equations are inconsistent, and their
span contains one. The same constant combination of their companions
gives \(P_D\), through \(2h+1\).
If the conclusion is empty, its value is one; the preceding formulas
still apply with no conclusion block. In the nonempty-zero-set case,
an implication to the empty clause forces all premise indicators to
be identically zero, so the already derived premise value is one.
\end{proof}

An arbitrary sound binary inference need not display a pivot. On the
falsifying flat of its conclusion, however, the two premises cannot hold
together. If neither premise alone implies the conclusion, their zero sets
cover that flat in two opposite fibers of an affine indicator. This recovers
a pivot from the semantic implication itself.

\begin{samepage}
\begin{lemma}[Binary affine separator]\label{lem:binary-separator}
\lean{claims/BinaryAffineZeroCover.lean}
Let \(V\) be a vector space over \(\F\), and \(A,B,D\) finite clauses
of affine true-indicators. Suppose \(A\wedge B\models D\), but neither
\(A\models D\) nor \(B\models D\). There is an indicator \(u\) already
in \(A\), taking both values on \(W=Z(D)\), such that
\[
 W\cap Z(A)=\{x\in W:u(x)=0\},\qquad
 W\cap Z(B)=\{x\in W:u(x)=1\}.
\]
\end{lemma}
\end{samepage}
\begin{proof}
Choose \(a\in W\) satisfying \(A\) and \(b\in W\) satisfying \(B\).
Soundness forces \(a\in Z(B)\) and \(b\in Z(A)\). Choose indicators
\(u\in A\), \(v\in B\) with \(u(a)=v(b)=1\); then
\(u(b)=v(a)=0\). Soundness also gives \(W\subseteq Z(A)\cup Z(B)\).
If \(x\in W\cap Z(B)\) had \(u(x)=0\), then \(v(x)=0\) as well.
The point \(a-b+x\) belongs to the affine flat \(W\), and the affine
identity \(g(a-b+x)=g(a)-g(b)+g(x)\) gives
\[
 u(a-b+x)=1-0+0=1,\qquad v(a-b+x)=0-1+0=1.
\]
It therefore satisfies both premises inside \(W\), a contradiction.
Thus \(u=1\) on \(W\cap Z(B)\), while \(u=0\) on \(W\cap Z(A)\).
The covering inclusion gives the two reverse implications and hence the
fiber equalities. The witnesses \(a,b\) give both values.
\end{proof}

\begin{samepage}
\begin{lemma}[Two-premise semantic inference]\label{lem:affine-cover}
\lean{claims/BinarySemanticAffineCover.lean}
Every sound inference \(A\wedge B\models D\)
between finite affine clauses admits a derivation using at most three
semantic weakening or complementary-parity resolution steps, introducing
at most two auxiliary clauses. It is either a one-premise weakening or
has the form
\[
 A\models D\vee(u=1),\qquad B\models D\vee(u=0),
 \qquad\frac{D\vee(u=1)\quad D\vee(u=0)}D,
\]
where \(u\) can be chosen from \(A\). Semantic implication is preserved
under any affine restriction of the underlying space.
\end{lemma}
\end{samepage}
\begin{proof}
If either premise implies \(D\), use weakening. Otherwise apply
Lemma~\ref{lem:binary-separator}: on \(Z(D)\), satisfaction of \(A\)
forces \(u=1\), and satisfaction of \(B\) forces \(u=0\).
This proves the two weakenings; resolution concludes. Restriction simply
precomposes the clauses with an affine map and preserves every implication.
\end{proof}

\begin{samepage}
\begin{lemma}[Clause compression]\label{lem:compression}
\lean{claims/AffineClauseCompression.lean}
On a finite-dimensional \(\F\)-vector space of dimension \(v\), every
finite affine clause has a semantically equivalent subclause of width at
most \(v+1\).
\end{lemma}
\end{samepage}
\begin{proof}
The space of affine forms has dimension \(v+1\). Extract a basis of the
span from the clause's own indicators. Simultaneous vanishing of the basis
is equivalent to vanishing of the entire span, hence gives the same
falsifying set. This works for tautologies as well as consistent systems
and introduces no new indicators.
\end{proof}

The lemmas above simulate one inference at a time, given blocks for its
premises and its conclusion. The polynomial calculus refutation needs all of
its axioms fixed in advance. Reserve a block for
every clause that can occur: one for each proof-line, two for the auxiliary
clauses that Lemma~\ref{lem:affine-cover} may introduce at that proof-line, and
one for each supplied initial clause. It then walks through the DAG in
topological order. Because each step uses only the already derived lines
\(P_A\), \(P_B\) and costs at most \(4h+1\), the degree of the whole
derivation is the maximum, not the sum, of the local costs.
A lemma may have a complicated upbringing; we do not have to pay for it
every time we use it.

\begin{samepage}
\begin{theorem}[A fixed registry for a finite proof DAG]\label{thm:dag-registry}
\lean{claims/AffineDAGRegistry.lean}
Consider an \(S\)-line affine-clause DAG on \(v\) old bits, using semantic
weakening, complementary resolution, or sound binary inference. Suppose
there is a finite supplied family \((A_j)_{j\in J}\), of widths at most
\(w\), such that each permitted initial clause is implied by some \(A_j\).
There is one complete accuracy-\(h\) registry with
\[
 N_{\mathrm{reg}}=3S+|J|\text{ slots},\qquad W=\max\{w,v+2\},
\]
at most \(N_{\mathrm{reg}}W\) inputs and at most \(v+hN_{\mathrm{reg}}W\) variables. If the supplied
initial values have PC derivations in this registry through
\(D\ge4h+1\), then every compressed source value has a derivation through
that same \(D\). An empty final clause yields a refutation.
\end{theorem}
\end{samepage}
\begin{proof}
Compress each source clause to width at most \(v+1\), preserving its
semantics, and choose the plans from Lemma~\ref{lem:affine-cover}.
Reserve one primary and two auxiliary slots per source proof-line, plus one
slot per supplied initial. An auxiliary clause adds one pivot indicator
to a primary clause, so has width at most \(v+2\). Unused slots can
contain the empty clause and contribute no inputs or coefficients.
Each input has \(h\) distinct coefficient variables, giving the inventory.
Fix these slots and all their complete axioms first. In a topological
order, simulate each planned weakening or resolution with
Lemmas~\ref{lem:clause-resolution} and \ref{lem:clause-weakening}.
Initial coverage is another weakening. Reuse completed premise values;
the degree ceiling never acquires a factor depending on proof height.
\end{proof}

Fresh unconstrained propositional variables may be fixed to zero before
this construction. Semantic implication survives restriction; a restricted
constant pivot is a semantic weakening. This does not authorize extension
axioms defining additional predicates.

\begin{samepage}
\begin{theorem}[Complete clause-to-PC transfer]\label{thm:clause-transfer}\lean{claims/BitPHPClauseTransfer.lean}
For arbitrary \(m,\ell\ge0\), a refutation of the usual bit-PHP CNF
on \(m\) rows and \(n=2^\ell\) labels with \(S\) proof-lines yields, for every
integer \(h\ge1\), a complete one-level old-affine family over
\(\mathcal Q_{m,\ell}\) with at most \(3S+\binom m2\) blocks and an ordinary-PC
refutation through
\begin{equation}\label{eq:simulation-degree}
 D=\max\{2h+\ell,4h+1\}.
\end{equation}
With \(v=m\ell\), \(N_{\mathrm{reg}}=3S+\binom m2\), and
\(W=\max\{\ell,v+2\}\), the total input count is at most \(N_{\mathrm{reg}}W\),
and the total variable count is at most \(v+hN_{\mathrm{reg}}W\). There is no proof-height restriction.
\end{theorem}
\end{samepage}
\begin{proof}
\lean{claims/BitPHPInitialBridge.lean} Introduce the compact initial linear clauses
\[
 C_{ii'}=\bigvee_{t=1}^\ell(b_{it}+b_{i't}=1).
\]
Each \(C_{ii'}\) semantically implies each label-specific initial clause
for the same pair in \eqref{eq:bit-php}: assigning the same
label to both rows falsifies \(C_{ii'}\).
Use these \(\binom m2\) clauses as the supplied family in
Theorem~\ref{thm:dag-registry}; their widths are at most \(\ell\).

For a compact initial clause set \(g_t=b_{it}+b_{i't}\) and
\(T_t=\prod_{s<t}(1-g_s)\).
Its value has the ordinary identity
\begin{equation}\label{eq:initial-value}
 P_{C_{ii'}}=E_{ii'}P_{C_{ii'}}
                 +\sum_{t=1}^{\ell}T_t(g_tP_{C_{ii'}}).
\end{equation}
The first term uses the old degree-\(\ell\) collision axiom;
the other terms use the block's companions.
Every multiple has degree at most \(2h+\ell\).
This is an NS, hence PC, derivation at that degree.

Apply Theorem~\ref{thm:dag-registry} in the source proof's acyclic order.
Initial values are obtained from \eqref{eq:initial-value} and
Lemma~\ref{lem:clause-weakening}. Each remaining inference is simulated
by Lemmas~\ref{lem:clause-resolution} and \ref{lem:clause-weakening}.
Theorem~\ref{thm:dag-registry} preserves the ceiling
\eqref{eq:simulation-degree}; the final empty clause has value one.
All blocks have inputs affine in the old \(v\) bits and disjoint fresh
coefficients. Their full companion and coefficient-domain families
are present.
\end{proof}

\pconlyfigure
\FloatBarrier

\section{Excluding one complete affine extension level}\label{sec:affine}

A small extension family can support a low-degree simulation only if its
new variables remain available. We will replace them by polynomials in the
original bits. Low rank permits an exact replacement; high rank permits a
different one, after multiplying every line by a common polynomial.

The common polynomial must satisfy many conditions without becoming zero.
High rank makes each restriction map small enough for this to be possible.
The bound on the number of blocks enters precisely at that dimension count.

\subsection{Affine linear algebra and ordinary restriction}

\begin{samepage}
\begin{lemma}[Affine spans and coordinates]\label{lem:affine-coordinates}
\leangroup{\leanref{claims/AffineSystemLinearAlgebra.lean}{linear algebra}
  \leansep \leanref{claims/FiniteAffineSpan.lean}{affine spans}
  \leansep \leanref{claims/FiniteAffineCoordinates.lean}{coordinates}}
Let a finite affine system on \(\F^v\) have span \(S\).
Its common zero set is nonempty exactly when \(1\notin S\).
In that case it is an affine flat of codimension \(r=\dim S\), every
affine polynomial vanishing on it belongs to \(S\), and any basis
\(F_1,\ldots,F_r\) of \(S\) can be completed to invertible affine
coordinates. Both coordinate substitutions preserve ordinary degree.
If the zero set is empty, a constant linear combination of the inputs is one.
\end{lemma}
\end{samepage}
\begin{proof}
Write the inputs as \(g_i=c_i+L_i\). If a relation among their linear
parts had nonzero constant part, their span would contain one. Otherwise
all relations respect the right sides \(-c_i\); linear algebra then solves
\(L_i(x)=-c_i\). Equivalently, define the consistent functional on the
span of the \(L_i\) and extend it to the full dual space.
When \(1\notin S\), independence of the \(F_j\) implies independence
of their linear parts: a relation would be a constant in \(S\), hence
zero, and then all coefficients vanish. Complete those linear parts to a
basis and retain the constants of the \(F_j\). This gives an invertible
affine coordinate change. The zero flat has first \(r\) coordinates zero,
so an affine polynomial vanishing there is a linear combination of those
coordinates. Substitution and its inverse have degree at most one,
which proves degree preservation. Inconsistency is exactly the failed
relation already identified, normalized to give one.
\end{proof}

\begin{samepage}
\begin{lemma}[Scalar cleanup and exact block degrees]\label{lem:cleanup}
\leangroup{\leanref{claims/FreshENSScalarCleanup.lean}{cleanup}
  \leansep \leanref{claims/FreshENSBlock.lean}{exact degrees}}
Over \(\F\), zero-span blocks, and unit-span blocks with \(h\ge1\), can
be removed by constant coefficient substitutions without increasing PC degree.
If a block has a genuine degree-one input, then
\(\deg P=2h\); each nonzero degree-one input has companion degree \(2h+1\).
In particular these degree equalities hold after discarding zero inputs in
a proper nonzero block. Without these hypotheses they are only upper bounds.
\end{lemma}
\end{samepage}
\begin{proof}
For zero span set all coefficients to zero. For unit span choose
\(\sum_i c_i g_i=1\), use these constants in the first coefficient row,
and use zero in the other rows. The product and companions become zero;
over \(\F\), every coefficient satisfies \(c_i^2=c_i\).
Apply the substitution lemma with degree multiplier one to the remaining
axioms. For the degree assertion, a genuine linear part contributes a
nonzero degree-two term to each factor, since its coefficient variables
are distinct. The polynomial ring is a domain, so degrees add in the
product and on multiplication by a degree-one input.
\end{proof}

\begin{samepage}
\begin{lemma}[Ordinary restriction dimension and ideal membership]
\label{lem:restriction-ideal}
\leangroup{\leanref{claims/OrdinaryRestrictionDimension.lean}{dimension}
  \leansep \leanref{claims/AffineRestrictionIdeal.lean}{ideal membership}}
An affine substitution into \(d\) variables maps a space
of polynomials of degree at most \(k\) to a space of dimension at most
\(\binom{d+k}{k}\). Suppose the common zero flat of a finite affine
system \(g_i\) is nonempty. If the ordinary polynomial restriction of
\(f\), of degree at most \(k\), to that flat is zero, then for \(k\ge1\)
\begin{equation}\label{eq:literal-ideal}
 f=\sum_i a_i g_i,\qquad \deg a_i\le k-1.
\end{equation}
Here restriction means substitution into affine free coordinates, not
merely evaluation at field points.
\end{lemma}
\end{samepage}
\begin{proof}
Substitution does not increase degree. The target monomials of degree at
most \(k\), counted by \(\binom{d+k}{k}\), span the image.
For ideal membership use the coordinates of Lemma~\ref{lem:affine-coordinates}.
A polynomial whose restriction at \(F_1=\cdots=F_r=0\) is literally zero
has each monomial divisible by at least one \(F_j\). Assign each monomial
to one such factor and remove it. The resulting cofactors have degree at
most \(k-1\). Pull back the affine coordinate change and express each
\(F_j\) as a constant combination of the original inputs.
\end{proof}

One polynomial must vanish under all the high-rank restrictions. For a
single restriction, its kernel may be large; intersecting many kernels can
still leave nothing. The local parity $x+y$ gives a concrete common kernel. At
$h_*=k_*=1$, consider two complete blocks with input tuples
\[
 G_1=(x+y,z,w),\qquad G_2=(x+y,z+1,w+1).
\]
Both have rank three, above the packing threshold two. Their zero flats
are $x=y,\ z=w=0$ and $x=y,\ z=w=1$. On
$V=\operatorname{span}\{1,x,y,z,w\}$, restriction gives the two maps
\begin{equation}\label{eq:local-restrictions}
\begin{aligned}
 R_1:\ a+bx+cy+dz+ew&\longmapsto a+(b+c)t,\\
 R_2:\ a+bx+cy+dz+ew&\longmapsto a+d+e+(b+c)t.
\end{aligned}
\end{equation}
Both images vanish precisely when $a=0$, $b=c$, and $d=e$. Hence
\[
 \ker R_1\cap\ker R_2=\operatorname{span}\{x+y,z+w\},
\]
and we can choose $f_*=x+y$. The two maps impose only three independent
conditions on a five-dimensional space. In the general argument the maps
are not so explicit. The relevant comparison is between the dimension of
$\mathcal L_k$ and the sum of all restriction-image dimensions. Restriction to a flat of codimension \(r\) is a linear map from
\(\mathcal L_k\) into polynomials of degree at most \(k\) in \(v-r\) free
coordinates, a space of dimension \(\binom{v-r+k}{k}\). For
\(r>3\ell(k+1)\) this is smaller than \(\dim\mathcal L_k\ge\binom mk\ell^k\)
by a factor \(e^{-k^2/m}\). If \(M\) such maps are given and
\(Me^{-k^2/m}<1\), their joint kernel is nonzero.

\begin{samepage}
\begin{lemma}[A common ordinary restriction kernel]\label{lem:kernel}
\leangroup{\leanref{claims/RestrictionKernelBounds.lean}{parameter bounds}
  \leansep \leanref{claims/CommonAffineRestrictionKernel.lean}{common kernel}}
Let \(m>0\), \(\ell\ge1\), \(h=3\ell\), \(M,k\ge1\), \(k\le m\), and
\[
 m\ln(4M)\le k^2.
\]
In a finite family of affine blocks in \(v=m\ell\) bits, suppose at most
\(M\) blocks are proper and have rank greater than \(h(k+1)\).
There is a nonzero \(f\in\mathcal L_k\) such that every block above
that rank threshold either has one in its input span or admits
\eqref{eq:literal-ideal}. In each proper high-rank block, the ordinary
restriction of \(f\) to its zero flat is zero.
\end{lemma}
\end{samepage}
\begin{proof}
Put \(r_*=3\ell(k+1)+1\): the extra one expresses the strict rank
inequality for integer ranks. Properness identifies rank with the
codimension of a nonempty zero flat (Lemma~\ref{lem:affine-coordinates}).
If \(r_*>v\), no proper high-rank block
exists and choose \(f=1\). Otherwise its restriction uses at most
\(d=v-r_*\) free coordinates. The dimension comparison uses
\[
 k!\binom{d+k}{k}\le(d+k)^k,\qquad
 k!\binom mk\ge(m-k+1)^k.
\]
Moreover \(d+k\le\ell(m-2k)\le\ell(m-k+1)(1-k/m)\).
The first inequality follows by inserting \(d=m\ell-3\ell(k+1)-1\)
and using \(\ell\ge1\); the second follows by expansion and \(k\le m\).
The numerator below bounds the full ordinary polynomial space of degree
at most \(k\) on each flat; the denominator counts just the degree-\(k\)
layer of the row-linear source space. For each restriction image,
\[
 \frac{\binom{d+k}{k}}{\binom mk\ell^k}
 \le\left(1-\frac{k}{m}\right)^k\le e^{-k^2/m}.
\]
The rank of the joint map is at most the sum of these image dimensions,
whether or not the restrictions are independent. That sum is at most
\[
 M e^{-k^2/m}\binom mk\ell^k
 \le\tfrac14\binom mk\ell^k<\dim\mathcal L_k.
\]
The joint restriction map has a nonzero kernel. Apply
Lemma~\ref{lem:restriction-ideal} to each proper high-rank block;
unit-span blocks need no restriction condition.
\end{proof}

\begin{samepage}
\begin{corollary}[A convenient kernel parameter]\label{cor:kernel-parameter}
\leangroup{\leanref{claims/RestrictionKernelBounds.lean}{parameter bounds}
  \leansep \leanref{claims/CommonAffineRestrictionKernel.lean}{common kernel}}
For \(m,M\ge1\), the choice \(k=\lceil\sqrt{m\ln(4M)}\rceil\) is
positive and satisfies the square condition of Lemma~\ref{lem:kernel}.
Thus that lemma applies whenever additionally \(k\le m\).
\end{corollary}
\end{samepage}
\begin{proof}
The radicand is positive; square the inequality defining the ceiling.
\end{proof}

\subsection{One simultaneous weighted substitution}

Fix the common $f$. Each block can now choose its own coefficient
substitution, but every substituted line carries that same weight.
The weighted axiom images must all have certificates below one degree
ceiling; otherwise the derivation could fail at a later use of an axiom.
The three cases below give that common ceiling, $kD+\deg f$.

\begin{samepage}
\begin{lemma}[Low-rank packing without a retained core]\label{lem:packing}
\lean{claims/LowRankENS.lean}
For affine inputs over \(\F\) of span rank \(r\le h(k+1)\), with
\(h,k\ge0\), there are coefficient polynomials of degree at most \(k\)
substituting the block product to
\[
 Z=\prod_{j=1}^r(1-F_j),
\]
where \(F_j\) is a basis of the input span. This polynomial has degree
at most \(r\), is the Boolean indicator of the common zero set, and
\(g_iZ\in\Ideal_{r+1}(\Bool)\) for every input. No properness assumption
is needed.
\end{lemma}
\end{samepage}
\begin{proof}
We partition the basis into \(h\) ordered bins of size at most \(k+1\).
For each bin \(J\), use
\[
 1-\prod_{j\in J}(1-F_j)
   =\sum_{j\in J} F_j\prod_{t\in J,\ t<j}(1-F_t).
\]
Express each \(F_j\) as a constant combination of the inputs. Every
coefficient has at most \(k\) affine factors. Multiplying the resulting
bin products gives \(Z\). At every Boolean point it is one precisely
when the whole input span vanishes, so \(g_iZ\) vanishes everywhere and
has degree at most \(r+1\). Apply Boolean reduction. If \(h=0\), the
rank bound forces \(r=0\), all inputs are zero, and the empty product is one.
\end{proof}

For the two blocks in \eqref{eq:local-restrictions}, the cofactor identity
is $f_*=1\cdot g_1$. Put coefficient one on the first input and zero on
the others. Both products become $1-f_*$. Their weighted companions have
the ordinary identities
\[
 f_*g_i(1-f_*)=-g_i(f_*^2-f_*),\qquad
 f_*^2-f_*=(x^2-x)+(y^2-y).
\]
The Boolean certificates have degree at most three. The low-rank block
$G_0=(x,y)$ instead uses the exact substitution from the overview,
with product $Z=(1-x)(1-y)$. For example,
\[
 f_*xZ=-f_*(1-y)(x^2-x),
\]
which has degree four. Both methods remove the coefficients; they do so
at different costs, with the same multiplier.

\begin{samepage}
\begin{lemma}[Weighted removal, including unit-span blocks]\label{lem:hybrid}
\lean{claims/AffineFamilyRemovalWithUnits.lean}
Let an old system \(\mathcal F\) in finitely many variables over \(\F\)
contain all old Boolean equations. Adjoin complete accuracy-\(h\) blocks
with finite affine input tuples and disjoint fresh coefficients, where
\(h,k\ge1\). Suppose this system has a degree-\(D\) ordinary-PC refutation,
\(D\ge2h+1\). Let \(f\) be any ordinary polynomial of degree \(s\le k\).
For every block of rank greater than \(h(k+1)\), suppose either its input
span contains one or it admits \eqref{eq:literal-ideal}. Then
\[
 f\in\Cons_{k(D+1)}(\mathcal F).
\]
No multilinearity or properness is required of \(f\) or of the blocks.
\end{lemma}
\end{samepage}
\begin{proof}
Make one simultaneous substitution fixing the old variables, with all
coefficient images of degree at most \(k\). There are three cases.
\begin{enumerate}
\item \emph{Low rank, \(r\le h(k+1)\).} For such a block use
Lemma~\ref{lem:packing}. Its weighted companions have Boolean certificates
through \(s+r+1\le s+h(k+1)+1\le s+kD\).
\item \emph{High rank, inconsistent inputs.} One belongs to the input
span; use the constant substitution making the product zero.
\item \emph{High rank, proper inputs.} Use the cofactors of
\eqref{eq:literal-ideal} in the first row and zero elsewhere. Its product
becomes \(1-f\), and a weighted companion becomes
\[
 fg_i(1-f)=-g_i(f^2-f).
\]
Its Boolean certificate costs at most \(2s+1\le s+kD\).
\end{enumerate}
A coefficient image \(\beta\) has weighted Boolean equation
\(f(\beta^2-\beta)\), certified through \(s+2k\le s+kD\).
An old axiom of degree \(e\le D\) has weighted image certified through
\(s+e\le s+kD\). These estimates initialize every source axiom used
by the refutation at the common ceiling \(kD+s\), irrespective of any
companion whose actual degree is below its general upper bound.
Lemma~\ref{lem:substitution}, with \(T=k\) and weight \(f\), replays the
whole primitive derivation. The final image of one is \(f\).
\end{proof}

\weightedonlyfigure

For the local family $G_0,G_1,G_2$ over $\mathcal Q_{17,4}$, all
coefficient images have degree at most one and $\deg f_*=1$.
Take the test degree $D_*=6$, the value of the general simulation bound
$\max\{2h_*+4,4h_*+1\}$ at $h_*=1$ and $\ell=4$.
If the augmented system had a PC
refutation through that degree, the lemma
would derive $f_*$ from the old base through $B_*=1(6+1)=7$.
Here $D_*\ge2h_*+1$ and the separator constructed in
Section~\ref{sec:local-separator} applies at that very degree.
It assigns $f_*$ both zero and one, so no such augmented refutation exists.

This local family has explicit cofactors. The registry of an arbitrary
assumed short proof need not have this form; there the size-dependent
kernel estimate must supply an arbitrary nonzero $f\in\mathcal L_k$.
The next theorem records the conditions that make its removal and
separation fit the same degree budget.

\begin{samepage}
\begin{theorem}[Finite-parameter affine exclusion]\label{thm:affine}
\lean{claims/AffineFamilyExclusion.lean}
Let \(m>0\), \(\ell\ge2\), \(n=2^\ell\), \(h=3\ell\), and let a finite
complete old-affine family be adjoined to \(\mathcal Q_{m,\ell}\).
Suppose \(M,k\ge1\), at most \(M\) blocks are proper and of rank greater
than \(h(k+1)\), and, with \(B=k(D+1)\),
\begin{equation}\label{eq:affine-room}
 m\ln(4M)\le k^2,\quad k\le m,\quad D\ge2h+1,\quad
 2B-1\le n,\quad4(k-1)<n.
\end{equation}
Then the augmented system has no ordinary-PC refutation through degree \(D\).
\end{theorem}
\end{samepage}
\begin{proof}
The common-kernel lemma supplies nonzero \(f\in\mathcal L_k\).
Weighted removal would derive \(f\) from the old base through \(B\),
contrary to Theorem~\ref{thm:row-separation}. Zero and unit spans are
already covered by removal, so no preliminary change of the family is needed.
\end{proof}

\begin{samepage}
\begin{lemma}[Uniform eventual parameter room]\label{lem:parameters}
\lean{claims/AffineParameterBounds.lean}
Fix natural numbers \(a,A,d\). Put \(n=2^\ell\), \(t=\ell+1\),
\(q=\sqrt2\), and \(k=\lceil t q^\ell\rceil\).
For every sufficiently large \(\ell\), simultaneously
\[
 \ell\ge2,\quad 1\le k\le n+1,\quad
 (n+1)\ln\bigl(4(A2^{a\ell}+1)\bigr)\le k^2,
\]
\[
 4(k-1)<n,\qquad 2k\bigl(A(\ell+1)^d+1\bigr)-1\le n.
\]
The threshold depends only on \(a,A,d\).
\end{lemma}
\end{samepage}
\begin{proof}
The ceiling gives \(tq^\ell\le k\le2tq^\ell\).
Set \(C=\ln(4(A+1))\). Since \(\ln2\le1\),
\(\ln(4(A2^{a\ell}+1))\le C+a\ell\).
For sufficiently large \(\ell\), \(2(C+a\ell)\le t^2\).
Together with \(n+1\le2n\) and \(q^{2\ell}=n\), this proves the
square condition. For every fixed natural \(e\),
\(t^e/q^\ell\longrightarrow0\): its successive ratio tends to
\(1/q<1\), so its tail is dominated by a geometric sequence.
Applying this to \(e=d+1\) gives eventually
\[
 4k(At^d+1)\le8(A+1)t^{d+1}q^\ell<n.
\]
This implies the remaining degree and packing bounds, and also \(k\le n+1\).
\end{proof}

\begin{samepage}
\begin{corollary}[Polynomial inventory and polylogarithmic degree]
\label{cor:polylog-exclusion}\lean{claims/AffineFamilyExclusion.lean}
Fix natural \(a,A,d\). For all sufficiently large \(\ell\), every finite
complete accuracy-\(3\ell\) affine family over \(\mathcal Q_{2^\ell+1,\ell}\)
with at most \(A2^{a\ell}+1\) proper blocks has no PC refutation of degree
\(D\le A(\ell+1)^d\).
\end{corollary}
\end{samepage}
\begin{proof}
Raise the proposed degree ceiling to \(D'=\max\{D,6\ell+1\}\).
It is at most \((A+7)(\ell+1)^{d+1}\). Apply
Lemma~\ref{lem:parameters} with constants \(a,A+7,d+1\), and use its
\(k\) and the upper inventory bound in Theorem~\ref{thm:affine}.
All required inequalities hold uniformly over the actual families.
\end{proof}

\section{The exponential lower bound}\label{sec:main-proof}

The simulation costs $D=12\ell+1$ when $h=3\ell$. Weighted removal
then costs $B=k(D+1)$, and the separator requires $B$ below roughly $n/2$.
This allows $k$ of order $n/\ell$. The common-kernel condition,
$(n+1)\ln(4M)\le k^2$, consequently permits $\ln S$ of order $n/\ell^2$.
The choice below leaves enough room for all the inequalities simultaneously.

\begin{proof}[Proof of Theorem~\ref{thm:main}]
Let \(\ell\ge32\), \(n=2^\ell\), \(m=n+1\), and put
\(x=n/(32768\,\ell^2)\). Suppose a refutation has \(S\le e^x\) proof-lines.
Theorem~\ref{thm:clause-transfer} with \(h=3\ell\) gives a complete affine
family over \(\mathcal Q_{m,\ell}\) with at most
\(M=3S+\binom{n+1}{2}\) blocks and a PC refutation of degree
\(D=12\ell+1\ge2h+1\). We verify \eqref{eq:affine-room} for
\[
 k=\left\lfloor\frac{n}{32\ell}\right\rfloor .
\]
An induction on \(\ell\ge32\) shows
\(32768\,(2\ell+4)\ell^2\le2^\ell\), that is, \(x\ge2\ell+4\). Hence
\(e^x\ge2^{2\ell+4}=16n^2\). Since \(\binom{n+1}2\le4n^2\),
\[
 4M\le12e^x+16n^2\le13e^x\le e^{2x},\qquad \ln(4M)\le2x .
\]
From \(n\ge64\ell\) we get \(k\ge1\) and \(k+1\le2k\), so
\(n<(k+1)\,32\ell\le64\ell k\). Therefore
\[
 (n+1)\ln(4M)\le2n\cdot2x=\frac{n^2}{8192\,\ell^2}
 \le\frac{n^2}{4096\,\ell^2}\le k^2 .
\]
Moreover \(32\ell k\le n\) gives \(k\le m\), \(4(k-1)<n\), and
\[
 2B=2k(12\ell+2)\le32\ell k\le n .
\]
Theorem~\ref{thm:affine} now says that the family has no PC refutation of
degree \(D\), a contradiction. Hence \(S>e^x\ge2^x\). Both rule conventions
are covered by Theorem~\ref{thm:clause-transfer}.
\end{proof}

The assumed short refutation would therefore yield a nonzero $f$ and a
functional $\Lambda=\lambda\circ\Delta\circ\tau$ for which
\[
 f\in\Cons_B(\mathcal Q_{m,\ell}),\qquad
 \Lambda(\Cons_B(\mathcal Q_{m,\ell}))=0,\qquad\Lambda(f)=1.
\]
The derivation and the separating construction disagree at the same final
polynomial (Figure~\ref{fig:contradiction}). Every suppressed line in the
derivation is still inside the degree bound; none can escape the functional
by passing through a higher-degree intermediate expression.

\contradictionfigure

The bound concerns DAG size alone: proof depth and repeated affine
pivots require no separate restriction.

Theorem~\ref{thm:main} also gives Corollary~\ref{cor:main-superpolynomial}:
$n/(32768\ell^2)\ge K\ell$ for every fixed $K$ and all sufficiently
large $\ell$. Appendix~\ref{sec:alternative-proof} retains the independent
argument using the smaller multiplier degree of Lemma~\ref{lem:parameters}.

\begin{samepage}
\begin{remark}[What the conclusion uses]\label{rem:conclusion-uses}
As explained in Section~\ref{sec:overview}, the contradiction is the
impossibility of a nonzero consequence in the separated row-linear space,
not a degree lower bound for deriving one. Equality of Nullstellensatz and
polynomial calculus consequences is used only for the old functional
system, and the argument bounds degree only: the number of blocks enters
through \(M\), while the monomial size of the simulating derivation is
irrelevant.
\end{remark}
\end{samepage}

\section{A generic sufficient condition}\label{sec:generic}

The simulation and weighted removal do not use the pigeonhole principle.
For another formula, the missing ingredient is a sufficiently large space of
polynomials that remains separated from low-degree consequences. Everything
in this section is over $\F$.

For a linear clause \(A\) with true-indicators \(g_1,\ldots,g_s\), its
\emph{clause polynomial} is \(\prod_i(1-g_i)\), of degree at most the width
\(s\); it vanishes exactly at the assignments satisfying \(A\). For a
polynomial system \(\mathcal F\), call a linear space \(U\) of polynomials
\emph{separated at degree \(B\)} if \(U\cap\Cons_B(\mathcal F)=\{0\}\).

\begin{samepage}
\begin{theorem}[Generic criterion]\label{thm:generic}
\leangroup{\leanref{claims/GenericAffineSubspaceConsequence.lean}{PC level}
  \leansep \leanref{claims/GenericCNFSubspaceCriterion.lean}{proof size}}
Let \(\mathcal I\) be a set of linear clauses in \(v\) variables, and let
\((A_j)_{j\in J}\) be a finite family of linear clauses of width at most
\(w\) such that every clause of \(\mathcal I\) is semantically implied by
some \(A_j\). Let \(\mathcal F\) be a polynomial system containing the
Boolean equations and the clause polynomials of all \(A_j\). Let
\(h,k\ge1\) with \(h(k+1)+1\le v\), put
\[
 D=\max\{2h+w,\,4h+1\},\qquad B=k(D+1),
\]
and let \(U\) be a space of polynomials of degree at most \(k\) that is
separated at degree \(B\). Then every \(\Res\) refutation of \(\mathcal I\)
with \(S\) proof-lines, under either rule convention of
Section~\ref{sec:introduction}, satisfies
\begin{equation}\label{eq:generic-criterion}
 \dim U\le\Bigl(3S+|J|\Bigr)\binom{v-h(k+1)-1+k}{k}.
\end{equation}
\end{theorem}
\end{samepage}
\begin{proof}[Proof by reference]
Theorem~\ref{thm:dag-registry} is already stated for arbitrary clause sets
with a supplied covering family; it gives a registry with \(3S+|J|\) slots.
The initial values are derived as in \eqref{eq:initial-value}, with the
clause polynomial of \(A_j\) in place of \(E_{ii'}\), through degree
\(2h+w\). This gives a PC refutation of degree \(D\) from \(\mathcal F\)
and the blocks. Lemma~\ref{lem:hybrid} is already stated for an arbitrary
old system. In place of Lemma~\ref{lem:kernel}, restrict \(U\) to the zero
flat of each proper block of rank above \(h(k+1)\): by
Lemma~\ref{lem:restriction-ideal} each image has dimension at most
\(\binom{v-h(k+1)-1+k}{k}\), so if \eqref{eq:generic-criterion} failed,
some nonzero \(f\in U\) would restrict to zero on all of them, and
Lemma~\ref{lem:hybrid} would give \(f\in\Cons_B(\mathcal F)\).
\end{proof}

\begin{samepage}
\begin{corollary}[Density form]\label{cor:generic-density}
\lean{claims/GenericCNFDensityBound.lean}
Under the hypotheses of Theorem~\ref{thm:generic},
\[
 3S+|J|\ \ge\ \delta\,\exp\!\left(\frac{hk^2}{v+k}\right),
 \qquad \delta=\frac{\dim U}{\binom{v+k}{k}} .
\]
\end{corollary}
\end{samepage}
\begin{proof}
Removing \(r\) of \(v\) variables shrinks \(\binom{v+k}{k}\) by at least the
factor \(\exp(-rk/(v+k))\), since each of the \(k\) factors
\((v-r+j)/(v+j)\) is at most \(1-r/(v+k)\). Take \(r=h(k+1)+1\ge hk\).
\end{proof}

Here \(\delta\) is the density of \(U\) among all polynomials of degree at
most \(k\). The bound is useful only if \(\delta\) is not much smaller than
\(\exp(-hk^2/(v+k))\) while \(U\) remains separated at degree
\(k(D+1)\). For bit PHP, \(U=\mathcal L_k\), \(w=\ell\), \(h=3\ell\), the old
system is \(\mathcal Q_{m,\ell}\), the covering family is the
\(\binom m2\) compact clauses of Theorem~\ref{thm:clause-transfer}, and
separation is Theorem~\ref{thm:row-separation}; the sharper count of
Lemma~\ref{lem:kernel} then gives Theorem~\ref{thm:main}.
Thus Sections~\ref{sec:clause-simulation} and~\ref{sec:affine} are generic
for any clause set whose covering clause polynomials are axioms of the base.
What is specific to bit PHP is the separation theorem over the compact base,
and the fact that a large separated space exists at all.

\begin{samepage}
\begin{remark}[A sufficient condition, not a tradeoff]\label{rem:not-tradeoff}
Theorem~\ref{thm:generic} is not a relation between \(\Res\) size and PC
refutation degree, nor between size and width. Its hypothesis is not that
\(\mathcal F\) has no low-degree refutation but that a large space \(U\) of
low-degree polynomials meets the low-degree consequences only in zero.
The conclusion of weighted removal is a nonzero derivable member of \(U\), not a
derivation of one, and nonzero polynomials can be legitimate low-degree
consequences: the satisfiable system \(\{x,\ x^2-x\}\) derives \(x\) and
does not derive one. We do not know a formula other than bit PHP for which
a separated space of useful density has been established.
\end{remark}
\end{samepage}

\begin{samepage}
\begin{remark}[Formulas with short refutations]\label{rem:short-proofs}
\lean{claims/ShortProofControl.lean}
Read contrapositively, a clause set with a size-\(S\) refutation admits no
separated space of degree \(k\) and dimension above the right side of
\eqref{eq:generic-criterion}. As a consistency check on the easiest parity
contradiction: if \(\mathcal F\) contains a polynomial \(L\) of degree at
most one and also \(1-L\), then it derives one at degree one, every
polynomial of degree at most \(B\) is derivable at degree \(B\ge1\), and
the only separated space is zero. The same happens, as it must, for Tseitin
contradictions with odd total charge, which have short \(\Res\)
refutations. Let \(d\) bound the maximum vertex degree. Over \(\F\),
the sum of the clause polynomials at a vertex is the affine polynomial
expressing that its constraint fails: both sides are multilinear and agree
on Boolean points. Summing over all vertices counts every edge variable
twice and leaves the total charge, which is one. Thus any base containing
the clause polynomials derives one at degree \(d\), and no nonzero space is separated
at any degree \(B\ge\max\{d,k\}\). Only the complementary-pair case and the
step from a derivation of one are formalized; the Tseitin computation is
not. This is a consistency check, not an audit of the known short \(\Res\)
refutations.
\end{remark}
\end{samepage}

\subsection{What does not follow}
\emph{Unary PHP.} The standard substitution from unary PHP into bit PHP
replaces a unary variable by the conjunction describing one bit label; see,
for example, \cite{BB21}. For bounded-depth Frege it has polynomial size
overhead and a constant depth increase. But an affine form in unary
variables becomes a polynomial in bit conjunctions, which is generally
nonlinear, so the substitution does not transfer Theorem~\ref{thm:main} to
unary PHP in \(\Res\). Nor does Theorem~\ref{thm:generic} apply directly:
no sufficiently dense separated space is known for the unary encoding.

\emph{Stronger systems.} A lower bound for the subsystem \(\Res\) does not
by itself give a Frege lower bound. In the translation of
\(\mathrm{AC}^0[p]\)-Frege proofs \cite{BIKPRS} the extension axioms are
nested: inputs of one block mention variables of others. Weighted removal uses that
all inputs are affine in the original variables, and we know of no
removal step for nested blocks.

Other possible applications suggested in review, such as further succinctly
encoded principles or resolution over linear equations modulo a prime
\(p>2\), have not been examined and nothing is claimed about them.

\section{Contributions, evidence, and the research framework}\label{sec:framework}

\subsection{Formal sources and reproducibility}
The main theorem, its superpolynomial consequence, the generic criterion,
and their mathematical dependencies have Lean~4 formalizations. The
per-statement links and Appendix~\ref{sec:verification} identify the
declarations at published commit
\href{https://github.com/kbr-/math-research/tree/8904bf09a5376f48a00d1c25d079bf0c6441502a}
{\texttt{8904bf09}}. The pinned project uses Lean~\texttt{4.34.0-rc2}
and a fixed Mathlib revision. The mathematical sources used here are
unchanged from that checkout.

The recorded verification includes a clean project build using the pinned
library caches, transitive axiom reports containing only
\texttt{propext}, \texttt{Classical.choice}, and \texttt{Quot.sound},
and a fresh kernel replay of the aggregate module
\texttt{claims.BitPHPPreprintRevision1} on 18 September 2026. This aggregate
imports the main results and their dependencies, including the chessboard
filling theorem. The verification appendix gives the exact checkout and
commands needed to reproduce those checks.

Earlier finite experiments checked primitive-PC simulations and matching
moment equations, including controls intended to fail. They supported
development of the argument; the asymptotic conclusion rests on the proof.
Formal verification establishes the encoded mathematical statements.
Interpretation of the proof-system definitions, literature comparisons,
and exposition remain separate matters.

\subsection{Noemesis}
The result was developed in \emph{Noemesis}, an open framework for autonomous
mathematical research maintained in the public \texttt{math-research}
repository and its online notebook~\cite{Repo}. The broader project seeks
a superpolynomial lower bound for ordinary PHP in every fixed-depth
$\mathrm{AC}^0[p]$-Frege system; that objective remains open.

Noemesis retains definitions, arguments, failed attempts, and explicit
claim statuses across assistant sessions. A restart guide and claim index
recover the relevant material; append-only research entries preserve
corrections and unsuccessful routes. Computation records retain commands,
complete outputs, resource use, and provenance. Separate formalization
work records the scope and dependencies of Lean proofs. Repeated process
reviews address missing tools and wasted work, while versioned sources and
the live notebook keep the development inspectable. These mechanisms make
the work recoverable; correctness still requires mathematical proof and
the appropriate checks.

The argument in this paper is self-contained. The notebook supplies its
development history and supporting experiments, rather than unexpanded
lemmas needed by the proof.

\subsection{Contributions and AI assistance}\label{sec:model-assistance}
The mathematics was produced by AI models under my direction. GPT-6 Astra
in Codex was the main contributor to the initial mathematical development,
Lean formalization, and drafting. Claude Sonnet~5, Claude Fable~5.1, and
Claude Opus~5 participated in earlier exploration. Claude Fable~5.1 later
formalized the exponential bound and generic criterion, delegated one
proof batch to Claude Opus~5, and revised the overview and related-work
discussion. The generic criterion was first extracted in its AI review
of the initial manuscript. GPT-6 in Codex prepared the September 2026
expository revision and its verification map.

The present revision was also prepared by GPT-6 in Codex, with my feedback
on voice, pacing, examples, and presentation. It rewrites the
exposition and adds a recurring hypothetical proof diagram, local algebraic
calculations, and a matching-moment filling. The local identities, parameter
conditions, and displayed moment data were checked separately. The theorem
statements and formal proofs are unchanged. The new illustrations are
explanatory examples, not additional machine-checked declarations.

My contribution is the design and refinement of the research framework,
the choice of problems, and the direction of the investigation; it is not
the mathematical derivations. My confidence in the correctness of the derivations rests on the Lean
formalization under the definitions in the linked sources. This
preprint has not been externally peer reviewed.

\begin{samepage}
The classical topological input is due to Bj\"orner, Lov\'asz, Vre\'cica,
and \v{Z}ivaljevi\'c~\cite{BLVZ94}; the product-extension method is due to
Buss, Impagliazzo, Kraj\'i\v{c}ek, Pudl\'ak, Razborov, and Sgall~\cite{BIKPRS}.
Original manuscript material is distributed under the
\href{https://creativecommons.org/licenses/by/4.0/}{Creative Commons
Attribution 4.0 International license}.
\end{samepage}

\appendix
\section{The required chessboard homology}\label{sec:topology}
\thispagestyle{plain}

A matching is a set of cells with no repeated row or column. Regard each
cell as a vertex and each matching as a simplex. Filling a cycle in this
complex supplies the next layer of matching moments in
Theorem~\ref{thm:moments}. The required range follows from the homological
consequence of the connectivity theorem of Bj\"orner, Lov\'asz, Vre\'cica,
and \v{Z}ivaljevi\'c~\cite[Theorem 1.1]{BLVZ94}.

The proof uses a cover by stars of first-row cells. Each star is a cone;
intersections are smaller chessboards. A finite chain argument combines
fillings in those intersections and in the nerve into a filling on the
whole board. All coefficients are in $\F$.

\subsection{Augmented chains and their elementary operations}

A finite augmented simplicial complex \(K\) is a downward-closed family
of subsets of a finite vertex set \(V\), containing the empty face.
It need not have a vertex. Write \(C_k(K)\) for coefficient functions
on its faces with exactly \(k\) vertices, extended by zero to all subsets
of \(V\). Thus the cell count \(k\) is reduced homological degree
\(k-1\): \(C_0(K)=\F\) is the augmentation group.
Write \(\operatorname{Exact}(K,k)\) when every \(k\)-cell cycle has a
\((k+1)\)-cell filler, and \(\operatorname{Acyclic}(K,n)\) when this
holds for every \(0\le k<n\). For \(n\ge1\), this says reduced homology
vanishes from degree \(-1\) through \(n-2\); for \(n=0\) it is vacuous.
In particular, \(\operatorname{Exact}(K,0)\) means that \(K\) has a vertex.

The index counts vertices rather than geometric dimension: a vertex is
in $C_1$, an edge in $C_2$, and a triangle in $C_3$. This makes the number
of vertices of a face equal to the degree of its matching monomial. The
augmentation group $C_0$ records the boundary of a vertex.

\begin{samepage}
\begin{proposition}[Supported augmented boundary]\label{prop:h01}
\lean{claims/AugmentedChains.lean}
For \(k\ge1\), the formula
\[
 (\partial c)(T)=\sum_{v\notin T}c(T\cup\{v\})
\]
defines a linear map \(C_k(K)\to C_{k-1}(K)\). Its restriction from
\(C_1(K)\) to \(C_0(K)\) is the sum of vertex coefficients.
The outgoing boundary of \(C_0(K)\) is zero.
The supported-coefficient model is linearly equivalent to functions on
exact-size faces, including for chessboard complexes.
\end{proposition}
\end{samepage}
\begin{proof}
A nonzero summand requires \(T\cup\{v\}\in K\) and cardinality \(k\).
Downward closure gives \(T\in K\), and \(v\notin T\) gives
\(|T|=k-1\). Finite sums are linear. At count zero every inserted face
is nonempty, so its coefficient vanishes. Restrict a supported function
to exact-size faces and invert by extension by zero. These maps are
linear and mutually inverse.
\end{proof}

\begin{samepage}
\begin{lemma}[The boundary squares to zero]\label{lem:h02}
\lean{third-party-claims/AugmentedBoundarySquared.lean}
The augmented boundary satisfies \(\partial^2=0\) at every cell count.
\end{lemma}
\end{samepage}
\begin{proof}
At a face \(T\), the double sum inserts two distinct vertices outside
\(T\). Pair the ordered insertion \((u,v)\) with \((v,u)\).
They contribute the same coefficient, twice, hence zero in \(\F\).
The outgoing augmentation boundary is already zero, so this includes
all low-count cases.
\end{proof}

\begin{samepage}
\begin{lemma}[Injective chain maps and relabeling]\label{lem:h03}
\leangroup{\leanref{claims/AugmentedChainMaps.lean}{chain maps}
  \leansep \leanref{claims/AugmentedSubcomplexRelabeling.lean}{relabeling}}
An injective vertex map carrying faces of \(K\) to faces of \(L\)
induces a degree-preserving linear pushforward commuting with boundary.
If it gives a bijection on the faces of the two complexes, this is a
chain isomorphism and transports cycle fillings in both directions.
In particular this applies to subcomplex inclusions, chessboard transpose,
and relabeling a residual board. It also identifies the supported boundary
with the boundary on exact-size chessboard faces.
\end{lemma}
\end{samepage}
\begin{proof}
On an image face \(f(S)\), give the pushforward coefficient \(c(S)\),
and assign zero off the range. Injectivity preserves cardinality and
makes the definition unambiguous. On a face outside the range all boundary
terms vanish. On an image face, insertions in the image correspond
bijectively to insertions in the source, and all other terms vanish.
Thus pushforward commutes with boundary. A face bijection supplies the
inverse coefficient map. Swapping row and column coordinates preserves
precisely the matching condition; residual relabeling has the same property.
Restriction to exact-size faces and extension by zero identify the two
boundary formulas term by term, including augmentation.
\end{proof}

\begin{samepage}
\begin{lemma}[Explicit cone contraction]\label{lem:h04}
\lean{claims/AugmentedCone.lean}
Suppose \(v\) is a cone vertex of \(K\), meaning that
\(S\in K\) implies \(S\cup\{v\}\in K\). Define
\[
 (s_vc)(U)=\begin{cases}c(U\setminus\{v\})&v\in U,\\0&v\notin U.\end{cases}
\]
Then \(s_v:C_k(K)\to C_{k+1}(K)\) and
\(\partial s_v+s_v\partial=\mathrm{id}\). Consequently
\(\operatorname{Exact}(K,k)\) holds for every \(k\ge0\).
\end{lemma}
\end{samepage}
\begin{proof}
The support assertion follows by adjoining the cone vertex.
If \(v\notin T\), the only nonzero term in \((\partial s_vc)(T)\)
is insertion of \(v\), giving \(c(T)\), and \((s_v\partial c)(T)=0\).
If \(v\in T\), all insertions outside \(T\) occur once in each of
\(\partial s_vc\) and \(s_v\partial c\) and cancel; the remaining
insertion of \(v\) into \(T\setminus\{v\}\) gives \(c(T)\).
For a cycle, the identity gives \(\partial(s_vc)=c\).
\end{proof}

\begin{samepage}
\begin{lemma}[The boundary of a simplex]\label{lem:h05}
\lean{third-party-claims/SimplexBoundary.lean}
The simplex on a nonempty finite vertex set is exact at every cell count.
If that set has size \(b\ge1\), its boundary complex, consisting of all
proper subsets, is exact at every \(k\) with \(k+1<b\).
Equivalently it is \(\operatorname{Acyclic}(\partial\Delta,n)\) whenever
\(n<b\). No top-dimensional vanishing is asserted.
\end{lemma}
\end{samepage}
\begin{proof}
The full simplex is a cone at any vertex. Contract a \(k\)-cycle in its
boundary inside the full simplex using Lemma~\ref{lem:h04}. Its filler
has \(k+1<b\) vertices on every supported face, so all those faces are
proper subsets and the filler already lies in the boundary complex.
\end{proof}

The filling in Figure~\ref{fig:moment-filling} is this cone operation:
adjoining $U$ to each edge of $AB+BK+KE$ gives the three triangles.
Its boundary consists of that path and the two edges from $U$ to the
path's endpoints. General matching-moment boundaries need not have a common
cone vertex, which is why the cover argument below is needed.

\subsection{The finite homological cover argument}

Filling inside each member of a cover is not enough: the pieces must
agree on their overlaps. The double-chain construction records both a face
and the cover members containing it. One boundary changes the face, the
other changes the set of cover members. Their compatibility lets us correct
one filling without losing the other.

Let \((L_i)_{i\in I}\) be finitely many subcomplexes of \(K\). Put
\(L_J=\bigcap_{i\in J}L_i\), with \(L_\varnothing=K\).
Its vertex-nonempty nerve is the complex
\[
 \mathcal N=\{\varnothing\}\cup
   \{J\subseteq I:L_J\text{ has at least one vertex}\}.
\]
An empty-face-only intersection does not give a nonempty nerve face.
The family covers \(K\) if every face belongs to some \(L_i\).

\begin{samepage}
\begin{lemma}[Cover double chains and horizontal exactness]\label{lem:h06}
\leangroup{\leanref{claims/FiniteComplexCover.lean}{cover definitions}
  \leansep \leanref{claims/CoverDoubleComplex.lean}{double chains}}
Let \(D_{a,b}\) consist of double coefficients \(d(J,S)\) supported where
\(|J|=a\), \(|S|=b\), \(S\in L_J\), and \(J\in\mathcal N\).
Insertion boundaries \(H\) in \(J\) and \(V\) in \(S\) lower the
respective cell counts, with zero outgoing maps at count zero, and satisfy
\[
 H^2=V^2=0,\qquad HV=VH,\qquad(H+V)^2=0.
\]
If the \(L_i\) cover \(K\), then for \(b>0\) every horizontally closed
\(d\in D_{a,b}\) has a horizontal filler in \(D_{a+1,b}\).
\end{lemma}
\end{samepage}
\begin{proof}
Removing an index preserves intersection membership and nerve membership.
Removing a vertex preserves membership in every complex, while the nerve
condition on \(J\) remains unchanged. These facts prove the support claims.
The square-zero identities are Lemma~\ref{lem:h02}; both \(HV\) and
\(VH\) sum the same coefficients \(d(J\cup\{i\},S\cup\{v\})\), so
interchanging finite sums proves commutation. Characteristic two gives
the total square-zero identity.

For each nonempty supported face \(S\), choose an index \(i(S)\) whose
member contains \(S\), and contract the index coefficient
\(J\mapsto d(J,S)\) at \(i(S)\), using the ambient formula of
Lemma~\ref{lem:h04}. Call the result \(u(J,S)\), and use zero on other
\(S\). Its nonzero coefficients have \(S\) in every indicated cover
member, including the newly chosen one. Because \(S\) contains a vertex,
the enlarged index set really belongs to the nerve. Thus
\(u\in D_{a+1,b}\), and the contraction identity gives \(Hu=d\).
The condition \(b>0\) is essential: the \(b=0\) edge is the nerve
chain complex and is not automatically exact.
\end{proof}

\begin{samepage}
\begin{lemma}[Homological cover lemma]\label{lem:nerve}
\lean{third-party-claims/HomologicalCover.lean}
Suppose \((L_i)\) covers \(K\). Let \(n\ge0\), and assume
\(\operatorname{Acyclic}(\mathcal N,n)\). Suppose also that for every
nonempty \(J\in\mathcal N\) and every \(r\ge0\) with \(|J|+r\le n\),
\(\operatorname{Exact}(L_J,r)\) holds. Then
\(\operatorname{Acyclic}(K,n)\).
In reduced degrees, putting \(n=q+2\), this is the usual assertion that
a \(q\)-acyclic nerve and \((q-|J|+1)\)-acyclic intersections give a
\(q\)-acyclic union, including the augmentation case \(q=-1\).
\end{lemma}
\end{samepage}
\begin{proof}
We use the finite double coefficients of Lemma~\ref{lem:h06}.
First note two filling facts. On the \(b=0\) edge, \(d(J,\varnothing)\)
is a nerve chain; nerve exactness fills horizontally at counts \(a<n\),
and the filler is vertically closed. For \(a>0\), a vertically closed
\(d\in D_{a,b}\) with \(a+b\le n\) fills vertically in \(D_{a,b+1}\):
for each \(J\) of size \(a\) use the assumed filling in \(L_J\), and
use zero for the other index sets. Nonempty \(J\) ensures that all filler
faces lie in \(K\) as well as the required intersection.

We claim that whenever \(a+b<n\) and \(Hd=Vd=0\), there is
\(u\in D_{a+1,b}\) with \(Hu=d\) and \(Vu=0\).
Induct on \(b\). The case \(b=0\) is the nerve-edge fact.
For \(b>0\), horizontal exactness first gives \(x\in D_{a+1,b}\)
with \(Hx=d\). Then \(Vx\in D_{a+1,b-1}\) is closed in both directions.
Induction gives \(y\in D_{a+2,b-1}\) with \(Hy=Vx\), \(Vy=0\).
The vertical filling fact gives \(z\in D_{a+2,b}\) with \(Vz=y\),
since \((a+2)+(b-1)=a+b+1\le n\).

\Needspace{10\baselineskip}
The terms occupy the following bidegrees:
\covercorrectiondiagram
Set $u=x+Hz$, in the same space as $x$. Then
\[
 Hu=Hx+H^2z=d,\qquad Vu=Vx+HVz=Vx+Hy=0.
\]
This proves the claim by finite induction.

Finally embed a \(k\)-cell cycle \(c\) of \(K\), \(k<n\), as
\(d(\varnothing,S)=c(S)\) in \(D_{0,k}\), with other coefficients zero.
It is closed in both directions. The claim gives \(u\in D_{1,k}\)
with \(Hu=d\), \(Vu=0\). The singleton-intersection filling gives
\(z\in D_{1,k+1}\) with \(Vz=u\), because \(1+k\le n\).
The chain \(f(S)=(Hz)(\varnothing,S)\) belongs to \(C_{k+1}(K)\), and
\(\partial f=(VHz)(\varnothing,\cdot)=(HVz)(\varnothing,\cdot)=c\).
Thus every required cycle fills, including \(k=0\).
\end{proof}

\subsection{Chessboard stars, arithmetic, and assembly}

The chessboard complex \(\Delta_{a,b}\) has vertex set \([a]\times[b]\)
and faces the matchings, namely sets with no repeated row or column.
The closed star of a vertex \(v\) consists of faces \(S\) such that
\(S\cup\{v\}\) is a face. It is a subcomplex and a cone at \(v\).

\begin{samepage}
\begin{lemma}[First-row star cover]\label{lem:h08}
\lean{third-party-claims/ChessboardStarCover.lean}
If \(a+1\le b\), the stars of \((1,j)\), \(j\in[b]\), cover
\(\Delta_{a+1,b}\).
\end{lemma}
\end{samepage}
\begin{proof}
A matching using row one belongs to the star of its cell in that row.
Otherwise it uses at most \(a<b\) columns, leaving some column \(j\)
free. It can then be extended by \((1,j)\).
\end{proof}

\begin{samepage}
\begin{lemma}[Intersections and residual-board chain maps]\label{lem:h09}
\lean{third-party-claims/ChessboardStarIntersections.lean}
The intersection of \(t\ge2\) distinct first-row stars in
\(\Delta_{a+1,b}\) consists exactly of matchings avoiding row one and
those \(t\) columns. It is chain-isomorphic to \(\Delta_{a,b-t}\),
including when \(t=b\) and the latter has only its empty face.
\end{lemma}
\end{samepage}
\begin{proof}
A matching in two different stars cannot use row one, since its row-one
cell would have to be both designated cells. It cannot use a designated
column in another row either, since extension by that star's cell would
fail. Conversely, avoiding all those rows and columns permits every
required extension. Relabel the remaining rows and columns and use
Lemma~\ref{lem:h03} to obtain the chain isomorphism.
\end{proof}

\begin{samepage}
\begin{lemma}[The star nerve and singleton fillings]\label{lem:h10}
\lean{third-party-claims/ChessboardStarNerve.lean}
For \(a\ge1\), \(b\ge2\), the nerve of the first-row stars of
\(\Delta_{a+1,b}\) is the boundary of the simplex on its \(b\) columns.
It is \(\operatorname{Acyclic}(\mathcal N,n)\) for every \(n<b\).
Every individual star is exact at every cell count.
\end{lemma}
\end{samepage}
\begin{proof}
A proper set of columns leaves an unused column and another row; their
cell belongs to every indicated star. The full set of stars has no
common vertex, by Lemma~\ref{lem:h09}. Thus the nerve is exactly the
proper-subset complex. Apply Lemma~\ref{lem:h05}; each singleton star
is a cone, so Lemma~\ref{lem:h04} fills its cycles.
\end{proof}

\begin{samepage}
\begin{lemma}[Induction parameters]\label{lem:h11}
\lean{claims/ChessboardParameterArithmetic.lean}
For natural \(a,b\), put
\[
 \nu(a,b)=\min\left\{a,b,\left\lfloor\frac{a+b+1}{3}\right\rfloor\right\}.
\]
It is symmetric, bounded by each side, zero if a side is zero, and at least
one when both sides are positive. For \(b\ge1\), \(\nu(1,b)=1\).
If \(2\le a\le b\), then \(\nu(a,b)\le b-1\), and for \(2\le t<b\),
\[
 \nu(a,b)+1\le\nu(a-1,b-t)+t,
 \qquad \min\{a-1,b-t\}<\min\{a,b\}.
\]
For \(s\ge1\), \(N\ge2s-1\) gives \(\nu(s,N)=s\).
Cell counts \(0\le k<n\) correspond exactly to reduced degrees
\(-1\le k-1\le n-2\), with subtraction interpreted in the integers.
\end{lemma}
\end{samepage}
\begin{proof}
Symmetry and the basic bounds follow directly from the minimum.
If \(a<b\), \(\nu\le a\le b-1\); if \(a=b\ge2\), then
\(\lfloor(2b+1)/3\rfloor\le b-1\).
For the recursive bound write \(\nu=\nu(a,b)\). Each smaller-board
entry is at least \(\nu-t+1\):
\[
 a-1\ge\nu-t+1,\qquad b-t\ge\nu-t+1,
 \qquad a+b-t\ge3\nu-1-t\ge3(\nu-t+1).
\]
The last inequality uses \(t\ge2\). Taking the smaller minimum proves
the bound; if the right side is negative the comparison is automatic.
The induction measure decreases because \(a-1<a=\min\{a,b\}\).
Finally \(N\ge2s-1\) implies \(N\ge s\) and
\(\lfloor(s+N+1)/3\rfloor\ge s\), proving the narrow-board identity.
The degree translation is the integer inequality obtained by subtracting one.
\end{proof}

\begin{samepage}
\begin{theorem}[Homological chessboard bound {\cite{BLVZ94}}]
\label{thm:chessboard}\lean{third-party-claims/ChessboardHomology.lean}
For \(a,b\ge1\), every \(k\)-cell cycle of \(\Delta_{a,b}\) has a
\((k+1)\)-cell filler whenever \(0\le k<\nu(a,b)\). Equivalently,
\(\widetilde H_j(\Delta_{a,b};\F)=0\) for
\(-1\le j\le\nu(a,b)-2\).
\end{theorem}
\end{samepage}
\begin{proof}
Induct strongly on \(\min\{a,b\}\), transposing by Lemma~\ref{lem:h03}
so that \(a\le b\). For \(a=1\), \(\nu=1\), and a vertex fills the
augmentation generator. For \(a\ge2\), cover by the first-row stars.
Their nerve is acyclic through count \(\nu\), since \(\nu\le b-1\).
Consider a nonempty nerve face \(J\) of size \(t\) and a count \(r\)
with \(t+r\le\nu\). For \(t=1\) the star is a cone. Otherwise
\(2\le t<b\), and its intersection is \(\Delta_{a-1,b-t}\), with both
sides positive. Write \(\mu=\nu(a-1,b-t)\).
Lemma~\ref{lem:h11} gives \(\nu+1\le\mu+t\), hence \(r<\mu\).
The smaller induction measure and the chain isomorphism of
Lemma~\ref{lem:h09} supply the required intersection filling.
Lemma~\ref{lem:nerve} now fills every \(k<\nu\) cycle of the original board.
\end{proof}

\begin{samepage}
\begin{corollary}[The matching-moment filling interface]
\label{cor:chessboard-filling}\lean{third-party-claims/ChessboardFillingProof.lean}
Let \(s\ge2\), \(N\ge2s-1\). Every cycle on exact-size
\((s-1)\)-faces of \(\Delta_{s,N}\) is the boundary of a chain on
exact-size \(s\)-faces. In reduced homology,
\[
 \widetilde H_{s-2}(\Delta_{s,N};\F)=0.
\]
For \(s=2\), the cycle hypothesis includes zero augmentation.
\end{corollary}
\end{samepage}
\begin{proof}
Lemma~\ref{lem:h11} gives \(\nu(s,N)=s\).
Apply Theorem~\ref{thm:chessboard} at cell count \(s-1<s\), then use the
exact-face boundary equivalence of Lemma~\ref{lem:h03}. This is precisely
the boundary equation used in Theorem~\ref{thm:moments}.
\end{proof}

\section{An independent proof of the superpolynomial consequence}\label{sec:alternative-proof}
The smaller choice of $k$ in Lemma~\ref{lem:parameters} already excludes
polynomial-size refutations. The argument below gives
Corollary~\ref{cor:main-superpolynomial} without using the explicit
exponential constant of Theorem~\ref{thm:main}.

\begin{proof}[Proof of Corollary~\ref{cor:main-superpolynomial}]
The corollary follows from Theorem~\ref{thm:main}, since
\(n/(32768\,\ell^2)\ge K\ell\) for large \(\ell\). There is also a direct
argument using the smaller \(k\) of Lemma~\ref{lem:parameters}, formalized
in the linked file.
Fix \(K>0\) and let \(a=\max\{\lceil K\rceil,2\}\).
Use Corollary~\ref{cor:polylog-exclusion} with parameters \(a,13,1\),
and take \(\ell\) above its uniform threshold and at least two.
Put \(n=2^\ell\). Suppose a refutation has \(S\le n^K\) proof-lines.
Since \(n\ge1\) and \(a\ge K\), \(S\le n^a\).
Theorem~\ref{thm:clause-transfer}, with \(h=3\ell\), constructs an actual
complete affine registry with
\[
 N_{\mathrm{reg}}\le3S+\binom{n+1}{2},\qquad D=12\ell+1.
\]
Its initial clauses are those of the usual bit CNF, and its empty final
clause gives a PC derivation of one from exactly the complete family over
\(\mathcal Q_{n+1,\ell}\). No additional initial-value hypothesis remains.
Now
\[
 \binom{n+1}{2}\le(n+1)^2\le4n^2\le4n^a,
 \qquad N_{\mathrm{reg}}\le7n^a\le13\,2^{a\ell}+1,
 \qquad D\le13(\ell+1).
\]
The number of proper blocks is at most \(N_{\mathrm{reg}}\), so these are precisely the
inventory and degree bounds excluded by the corollary, a contradiction.
The threshold is independent of the alleged proof. Therefore \(S>n^K\).

This argument uses \(k=\lceil(\ell+1)(\sqrt2)^\ell\rceil\), giving
\begin{equation}\label{eq:main-asymptotics}
 k=O(\sqrt n\log n),\qquad
 k(D+1)=O\bigl(\sqrt n(\log n)^2\bigr)=o(n).
\end{equation}
This is far below the room \(k\approx n/(32\ell)\) used for Theorem~\ref{thm:main}.
\end{proof}

\section{Formal-verification map and reproduction}\label{sec:verification}

The mathematical arguments are given in the body and
Appendices~\ref{sec:topology} and~\ref{sec:alternative-proof};
reading them does not require Lean or a supplementary document. This appendix
locates the principal formal declarations and gives a complete command sequence
for checking them. Formal verification establishes the statements under the
encoded definitions; external review of those definitions and of the literature
claims remains separate.

\begin{samepage}
\begin{remark}[Generality]\label{rem:generality}
Several statements in Sections~\ref{sec:algebra}--\ref{sec:affine}
were formalized in greater generality than is used here: over an arbitrary
field, for infinite variable sets or infinite-dimensional spaces, or in
arbitrary commutative rings. The linked Lean files record the exact
hypotheses. One distinction made visible by formalization is used in the main argument:
the degree of a block product or companion is in general only bounded above
by the stated value, with equality under the hypotheses of
Lemma~\ref{lem:cleanup}. The argument only uses the upper bounds.
\end{remark}
\end{samepage}

\subsection{From paper statements to formal declarations}

The table uses the pinned [Lean] links of the paper. File paths are relative to
\texttt{formalization/}; all names are in namespace \texttt{MathResearch}.

\begin{samepage}
The prefix \texttt{PC.} below abbreviates
\mbox{\texttt{MathResearch.PolynomialCalculus.}}, and \texttt{TP.} abbreviates
\mbox{\texttt{MathResearch.ThirdParty.}}; unprefixed names have prefix
\texttt{MathResearch.}
\end{samepage}
Each file's header specifies its scope and dependencies.
The labels in Appendix~\ref{sec:topology} correspond to the H01--H13 route in the linked
formalization. The table selects principal declarations rather than every helper lemma;
the per-statement links in the body cover the remaining steps.

\begingroup
\small
\renewcommand{\arraystretch}{1.3}
\begin{longtable}{@{}>{\raggedright\arraybackslash}p{0.30\textwidth}>{\raggedright\arraybackslash}p{0.65\textwidth}@{}}
\toprule
Paper statement and scope & Lean file and declaration\\
\midrule
\endfirsthead
\toprule
Paper statement and scope & Lean file and declaration\\
\midrule
\endhead
\bottomrule
\endfoot
Matching extension over \(\F\), Theorem~\ref{thm:moments}
& \leanref{claims/MatchingMomentExtension.lean}{\nolinkurl{claims/MatchingMomentExtension.lean}}\newline
\nolinkurl{PC.matching_moments_extend}\\
Chessboard filling over \(\F\), Corollary~\ref{cor:chessboard-filling}
& \leanref{third-party-claims/ChessboardFillingProof.lean}{\nolinkurl{third-party-claims/ChessboardFillingProof.lean}}\newline
\nolinkurl{TP.chessboard_filling}\\
Cube degree drop, Lemma~\ref{lem:cube-drop}
& \leanref{claims/BinaryCubeDegreeDrop.lean}{\nolinkurl{claims/BinaryCubeDegreeDrop.lean}}\newline
\nolinkurl{PC.cubeDifference_degree}\\
Coefficient isolation, Lemma~\ref{lem:row-cube}
& \leanref{claims/RowCubeCoefficientIsolation.lean}{\nolinkurl{claims/RowCubeCoefficientIsolation.lean}}\newline
\nolinkurl{PC.row_cube_coefficient_isolation}\\
Old-system separator, Theorem~\ref{thm:row-separation}
& \leanref{claims/CubeResidualDualSeparation.lean}{\nolinkurl{claims/CubeResidualDualSeparation.lean}}\newline
\nolinkurl{PC.cube_residual_dual_separation}\\
Bounded cofactors for ordinary restriction, Lemma~\ref{lem:restriction-ideal}
& \leanref{claims/AffineRestrictionIdeal.lean}{\nolinkurl{claims/AffineRestrictionIdeal.lean}}\newline
\nolinkurl{affine_restriction_coefficients}\\
Common restriction kernel, Lemma~\ref{lem:kernel}
& \leanref{claims/CommonAffineRestrictionKernel.lean}{\nolinkurl{claims/CommonAffineRestrictionKernel.lean}}\newline
\nolinkurl{common_affine_kernel_of_square}\\
Removal including unit-span blocks, Lemma~\ref{lem:hybrid}
& \leanref{claims/AffineFamilyRemovalWithUnits.lean}{\nolinkurl{claims/AffineFamilyRemovalWithUnits.lean}}\newline
\nolinkurl{affine_family_removal_with_units}\\
Finite exclusion with the square condition, Theorem~\ref{thm:affine}
& \leanref{claims/AffineFamilyExclusion.lean}{\nolinkurl{claims/AffineFamilyExclusion.lean}}\newline
\nolinkurl{affine_family_exclusion_of_square}\\
Actual DAG registry, Theorem~\ref{thm:dag-registry}
& \leanref{claims/AffineDAGRegistry.lean}{\nolinkurl{claims/AffineDAGRegistry.lean}}\newline
\nolinkurl{PC.affine_dag_registry}\\
Usual CNF-to-PC transfer, Theorem~\ref{thm:clause-transfer}
& \leanref{claims/BitPHPClauseTransfer.lean}{\nolinkurl{claims/BitPHPClauseTransfer.lean}}\newline
\nolinkurl{PC.usual_bitPHP_PC_transfer}\\
Main bound, both rule conventions, Theorem~\ref{thm:main}
& \leanref{claims/BitPHPExponential.lean}{\nolinkurl{claims/BitPHPExponential.lean}}\newline
\nolinkurl{bitPHP_exponential}; \nolinkurl{bitPHP_exponential_base_two}\\
Generic sufficient condition over \(\F\), Theorem~\ref{thm:generic}
& \leanref{claims/GenericCNFSubspaceCriterion.lean}{\nolinkurl{claims/GenericCNFSubspaceCriterion.lean}}\newline
\nolinkurl{PC.generic_CNF_subspace_criterion}\\
Density form, Corollary~\ref{cor:generic-density}
& \leanref{claims/GenericCNFDensityBound.lean}{\nolinkurl{claims/GenericCNFDensityBound.lean}}\newline
\nolinkurl{PC.generic_CNF_density_bound_simple}\\
\end{longtable}
\endgroup

The homological argument is proved in Lean, including the chessboard filling
used by matching extension; it is not imported as an unchecked axiom.
The final transitive axiom checks described below cover these supporting
results as well as the headline theorem.

\subsection{A pinned checkout and verification commands}

Use a machine with Git, Python 3, a C toolchain and \texttt{make}, and
\href{https://lean-lang.org/install/manual/}{elan}, Lean's toolchain manager,
on the command path. Network access is needed for the checkout, the pinned
Lean toolchain and Mathlib dependencies, and their compiled library cache.
No Python packages are required. Allow disk space for the Lean toolchain and
library build artifacts. The project pins Lean \texttt{4.34.0-rc2};
\texttt{lakefile.lean} and \texttt{lake-manifest.json} pin Mathlib and its
transitive dependencies. Do not update these pins.

The following shell commands create a fresh checkout at commit
\href{https://github.com/kbr-/math-research/tree/8904bf09a5376f48a00d1c25d079bf0c6441502a}
{\texttt{8904bf09}}, which contains the aggregate verification module and the
theorem sources used by every [Lean] link in this paper. The present revision changes exposition and illustrations, so it
uses the same formal verification target, \texttt{BitPHPPreprintRevision1}.

\begingroup\small
\begin{verbatim}
git clone https://github.com/kbr-/math-research.git bit-php-check
cd bit-php-check
git checkout --detach 8904bf09a5376f48a00d1c25d079bf0c6441502a
cd formalization
export MATHLIB_NO_CACHE_ON_UPDATE=1
export LEAN_NUM_THREADS=2
lean --version
lake exe cache get claims/BitPHPPreprintRevision1.lean
python3 verify.py --target claims/BitPHPPreprintRevision1.lean \
  --out verification.txt
lake env leanchecker --fresh claims.BitPHPPreprintRevision1
\end{verbatim}
\endgroup

On a fresh checkout the cache command retrieves the required library artifacts;
the verifier builds the project proofs and prints the types and transitive
axioms of the aggregate's listed declarations. It rejects unfinished proofs
and nonstandard axioms. The output file is written in \texttt{formalization/}; use
a new name if repeating the command. All commands must exit successfully.
The final command replays the aggregate and its imports in a fresh Lean kernel
environment. A successful build alone does not substitute for the type and
axiom checks or that replay.

For example, the expected main-theorem axiom report is
\begin{quote}\small\ttfamily
'MathResearch.bitPHP\_exponential' depends on axioms:\par
[propext, Classical.choice, Quot.sound]
\end{quote}
These are Lean's standard foundations; neither \texttt{sorryAx} nor a
custom mathematical axiom is allowed. The aggregate also checks the
superpolynomial corollary, both exponential forms, the generic criterion
and density form, and the complementary-pair short-proof control, with their
imports. It does not formalize the informal Tseitin remark, the literature
comparison, or the AI-assistance disclosure.

\end{document}